\documentclass{ws-ijcga}

\usepackage{mathrsfs}

\begin{document}

\markboth{Fukuda et al.}
{Polyominoes and Polyiamonds as Fundamental Domains of Isohedral Tilings
}

\catchline

\title{
POLYOMINOES AND POLYIAMONDS AS FUNDAMENTAL DOMAINS OF ISOHEDRAL TILINGS
WITH ROTATIONAL SYMMETRY
(Submitted to International Journal of Computational Geometry and Applications (IJCGA))
}

\author{HIROSHI FUKUDA
}

\address{
College of Liberal Arts and Sciences, Kitasato University, 1-15-1 Kitasato, Sagamihara,
Kanagawa 252-0373, JAPAN\\
fukuda@kitasato-u.ac.jp
}

\author{CHIAKI KANOMATA}

\address{
School of Administration and Informatics, University of Shizuoka, 52-1 Yada, Shizuoka 422-8526 JAPAN
}

\author{NOBUAKI MUTOH}

\address{
School of Administration and Informatics, University of Shizuoka, 52-1 Yada, Shizuoka 422-8526 JAPAN\\
muto@u-shizuoka-ken.ac.jp
}

\author{GISAKU NAKAMURA}

\address{
School of Administration and Informatics, University of Shizuoka, 52-1 Yada, Shizuoka 422-8526 JAPAN
}

\author{DORIS SCHATTSCHNEIDER}

\address{
Mathematics Dept. PPHAC Moravian College, 1200 Main St. Bethlehem, PA 18018-6650\\
schattdo@moravian.edu
}

\maketitle

\pub{Received (received date)}{Revised (revised date)}{Communicated by (Name)}

\begin{abstract}
We describe computer algorithms that produce the complete set of isohedral tilings by
$n$-omino or $n$-iamond tiles in which the tiles are fundamental domains and the tilings have 3-, 4-,
or 6-fold rotational symmetry. 
The symmetry groups of such tilings are of types {\bf p3}, {\bf p31m}, {\bf p4},
{\bf p4g}, and {\bf p6}. 
There are no isohedral tilings with {\bf p3m1}, {\bf p4m}, or {\bf p6m} symmetry groups that have
polyominoes or polyiamonds as fundamental domains. 
We display the algorithms' output and give
enumeration tables for small values of $n$. 
This expands earlier works.\cite{fukuda2006,fukuda2008} 
%
%
\end{abstract}

\section{Introduction}	
\label{intro}
Polyominoes and polyiamonds and their tiling properties have been the subject of 
computational geometry research that investigated 
which polyominoes can tile the plane isohedrally and 
which can tile by translations alone. \cite{rhoads,Kev}. 
In earlier papers\cite{fukuda2006,fukuda2008}, 
we gave algorithms to create polyominoes and
polyiamonds that were fundamental domains for isohedral tilings having {\bf p3}, {\bf p4}, or {\bf p6} symmetry
groups. 
In this paper, we consider the expanded task of producing polyomino and polyiamond
tiles that generate isohedral tilings of types {\bf p3}, {\bf p3m1}, {\bf p31m}, {\bf p4}, {\bf p4m}, {\bf p4g}, {\bf p6} 
or {\bf p6m} and for
which the tiles are fundamental domains of the tiling.
Recently an extension of this study was published as a separate paper\cite{fukuda2011} 
in which we carried out these investigations for symmetry groups of
types {\bf pmm}, {\bf pmg}, {\bf pgg}, and {\bf cmm}.

A polyomino (or $n$-omino) is a tile homeomorphic to a disk, made up of $n$ unit squares that
are connected at their edges; that is, the intersection of two unit squares in the polyomino is
either empty or an edge of both squares. 
Similarly a polyiamond (or $n$-iamond) is a tile
homeomorphic to a disk, made up of $n$ unit equilateral triangles that are connected at their edges;
the intersection of two unit triangles in the polyiamond is either empty or an edge of both
triangles.

An isohedral tiling of the plane is a tiling by congruent tiles in which the symmetry group of
the tiling acts transitively on the tiles. A fundamental domain for an isohedral tiling is a region
of least area that generates the whole tiling when acted on by the symmetry group of the tiling.
A fundamental domain for an isohedral tiling cannot contain two points that are identical under
the action of the symmetry group of the tiling. 
This leads immediately to the following observations:

\begin{romanlist}
\item 
{\it In an isohedral tiling in which each tile is a fundamental domain, no tile can contain a
rotation center or axis of reflection for the whole tiling except on its boundary.}
\item 
{\it In an isohedral tiling having a polyomino as fundamental domain, a 4-fold center of the tiling
can only occur at a vertex of a square that is a ``corner" of the polyomino; 
that is, only one unit square of the polyomino contains that vertex.}
\item
{\it In an isohedral tiling having a polyiamond as fundamental domain, 3-fold or 6-fold centers of
the tiling can only occur at vertices of unit triangles on the boundary of the polyiamond.
Moreover, a 3-fold center can only occur at a vertex at which exactly two unit triangles meet. A
6-fold vertex can only occur at a vertex that is a ``corner" of the tile, that is, exactly one unit
triangle of the polyiamond contains that vertex.}
\item
{\it There are no 3-or 6-fold centers in a tiling by polyominoes. 
Hence, there are no \mbox{\bf p3}, \mbox{\bf p3m1},
\mbox{\bf p31m}, \mbox{\bf p6} or \mbox{\bf p6m} isohedral tilings by polyominoes. 
There are no 4-fold centers in a tiling by
polyiamonds. Hence, there are no \mbox{\bf p4}, \mbox{\bf p4m}, or \mbox{\bf p4g} isohedral tilings by polyiamonds. }
\end{romanlist}

In each of the sections that follow, we begin with a fixed lattice of symmetry elements for a
symmetry group $G$ (that is, a fixed array of rotation centers, reflection axes and glide-reflection
axes for elements of $G$) and give a backtracking procedure to produce a complete set of
polyominoes (or polyiamonds) that are fundamental domains for $G$. 
A tile $T$ is {\it a fundamental domain for a symmetry group $G$} 
if the action of $G$ on $T$ produces an isohedral tiling and $T$ is a
region of minimal area for which $G$ can generate that tiling. $G$ will be contained in (or equal to)
the full symmetry group of the tiling.

As we consider symmetries of our isohedral tilings, the following theorem will be useful. \cite{fukuda2011}

\begin{theorem}
\label{th1}
Let $G$ be one of the 17 two-dimensional symmetry groups and ${\cal T}$ an isohedral tiling
generated by $G$ acting on a tile $T$ that is a fundamental domain for $G$. 
Let $G'$ be the full
symmetry group of ${\cal T}$. 
If $G$ is a proper subgroup of $G'$, there is element of $G'$ (other than the
identity) that leaves $T$ invariant. In this case, a fundamental domain for $G'$ has area smaller
than $T$.
\end{theorem}
%

\section{{\bf p4}}
\label{p4}
\subsection{Creating polyominoes as fundamental domains for {\bf p4} symmetry groups}
\label{p4:create}
We begin with a lattice of unit squares as shown in Fig.~\ref{fig:p4lattice}. 
Our $n$-omino tiles that are fundamental domains for a {\bf p4} group will be built from these unit squares. 
By observation (ii) in section \ref{intro}, 
the 4-fold rotation centers for a {\bf p4} symmetry group must be located at lattice points.
So, first we place a 4-fold rotation center, (a black circle) at a lattice point and call this the origin.
We then place orthogonal unit vectors $\mathbf u$ and $\mathbf v$ at the origin (see Fig.~\ref{fig:p4lattice}). 
Next we place a second
4-fold rotation center, the white circle in Fig.~\ref{fig:p4lattice}, at $x{\mathbf u} + y{\mathbf v}$, 
where $x$ and $y$ are nonnegative
integers, not both $0$. These two choices of 4-fold rotation centers determine a whole {\bf p4} lattice of
rotation centers; rotations about the two placed black and white centers generate a {\bf p4} group $G$.
\begin{figure}[h]
\centerline{\psfig{file=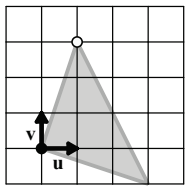,width=2.7cm}}
\caption{
A lattice of unit squares with black and white 4-fold rotation centers that generate a {\bf p4} group. 
The shaded isosceles right triangle is a standard fundamental domain; 
its unmarked third vertex is a 4-fold center equivalent to
the white center. Here $x = 1$, $y = 3$; the area of the fundamental domain is $5$ sq. units.
\label{fig:p4lattice}
}
\end{figure}

The area $S$ of a fundamental domain for this {\bf p4} group $G$ is given by
\begin{equation}
S=\frac{x^2+y^2}{2}
\end{equation}
taking the area of a unit square as $1$. 
Since we want our $n$-omino to be a fundamental domain, $n = S$ and $S$ must be an integer. 
Therefore, $x$ and $y$ must have the same parity, and
\begin{equation}
n=1, 2, 4, 5, 8, 9, 10, 13, 16, 17, 18, 20, 25,	25,	26,	29,	32,	34,	36,	37,	40,	41,	45,	49,	
\ldots
\end{equation}
where two pairs $(x,y)$, namely, $(5,5)$ and $(7,1)$ correspond to $n = 25$.
All unit squares in the lattice
are classified into $n$ equivalence classes by the action of the {\bf p4} group $G$. 
We denote the equivalence class of a unit square $e$ as $C(e)$.

We construct a set $\mathscr{T}_n$ of $n$-ominoes
that are fundamental domains for the {\bf p4} group $G$ and a
given $n$ by following Procedure 1 below, using these definitions:

\begin{enumerate}
\item
$T$ is a set of unit squares; $B(T)$ is a set of unit squares that are edge-adjacent to the squares
in $T$; $\mathscr{T}_n$ is a set of $n$-ominoes.
\item
When $T$ is the empty set, we define $B(\emptyset)$ as the set of four squares around the origin (the
four unit squares at the lower left in Fig.~\ref{fig:p4lattice}.)
\item
We define $E(T)$, the boolean function of $T$, which is true if $\# T = n$ and the white circle is on
the boundary of $T$. Otherwise $E(T)$ is false.
\item
$ B'(T)=\{e| e \in B(T), C(e)\ne C(f), \mbox{for all} f \in T \} $. 
This is the set of all unit squares that are
edge-adjacent to those in $T$, but not equivalent to any unit squares in $T$.
\end{enumerate}

Procedure 1 creates a sequence of pairs of sets $(T, U_T)$, 
in which $U_T$ is the set of unit squares
that can be added to $T$ to create the next set $T$ in the sequence. 
The sets $T$ are built up, one
element at a time, until a set $T$ is achieved for which $E(T)$ is true, 
at which time $T$ is added to the
collection $\mathscr{T}_n$. 
When this takes place, the procedure backtracks to the most recent previous pair
$(T, U_T)$ for which $U_T \ne \emptyset$, and repeats the process.

\vspace{1em}
\noindent
{\bf Procedure 1.}
\begin{enumerate}
\item
Fix $n$ (from the list in (2) above). Begin with $\mathscr{T}_n$ empty.
\item
Make $T$ empty. Make $U_T = B'(T) = B(\emptyset)$. Make $k = 0$. Make $S_k = \{(T, U_T)\}$.
\item
For $(T, U_T)$ in $S_k$, if $U_T \ne \emptyset$, 
choose an element $e$ (a unit square) of $U_T$. Remove $e$ from $U_T$
and save the pair $(T, U_T)$ in $S_k$.
\item
Increase $k$ by $1$. Add $e$ (from step 3) to $T$ to create a new $k+1$-omino $T$, 
and let $U_T = B'(T)$ for the new $T$. Add the new pair $(T, U_T)$ to $S_k$.
\item
If $T = \emptyset$ and $U_T = \emptyset$, quit the procedure.
\item
If $E(T)$ is true, an $n$-omino tile is completed and add $T$ to $\mathscr{T}_n$ provided there is no equivalent
tile in $\mathscr{T}_n$. We regard two $n$-ominoes equivalent if conditions (a) and (b) below are satisfied.
\begin{enumerate}
\item
The tiles are congruent (including mirror reflection).
\label{proc1a}
\item
When the tiles are superimposed, the positions of $m$-fold rotational centers on the
boundaries of the tiles are the same. 
If there are several $m$-fold rotation centers for the same $m$, 
they can be permuted appropriately before comparison.
\label{proc1b}
\end{enumerate}
\item
If $U_T = \emptyset$, decrease $k$ by $1$.
\item
Go back to step 3.
\end{enumerate}

Fig.~\ref{p4fighi} shows the set of inequivalent $n$-ominoes in $\mathscr{T}_n$
for $n \le 8$. 
From each of these $n$-ominoes
we can obtain the associated {\bf p4} tiling by using the black and white circles as 4-fold rotation
centers. 
Fig.~\ref{p4tiling} shows the corresponding isohedral tilings produced by n-ominoes in Fig.~\ref{p4fighi} for $n \le 5$.
\begin{figure}[h]
\centerline{\psfig{file=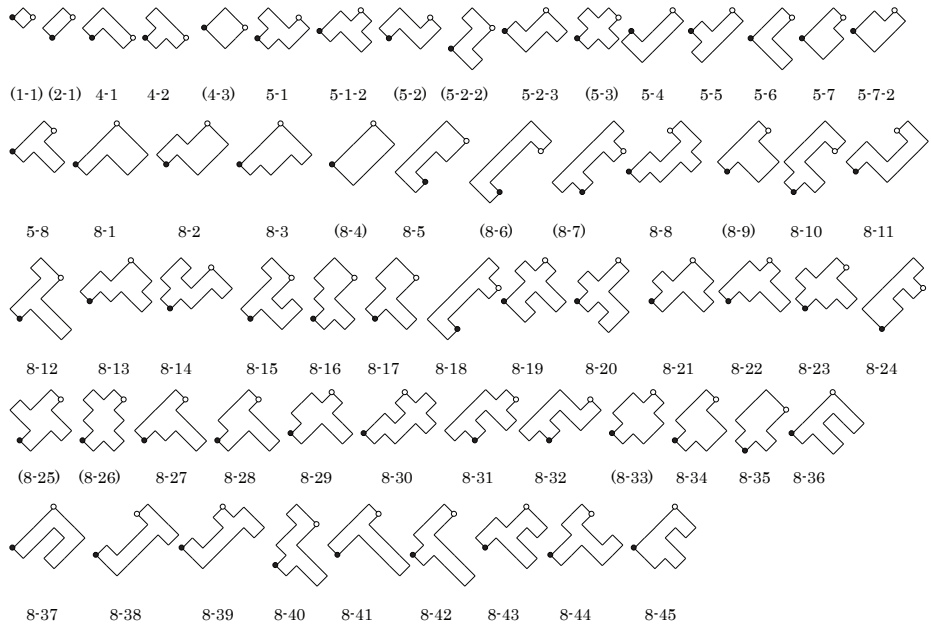,width=\linewidth}}
\caption{
List of $n$-ominoes produced by the procedure in section 2.1 for $n \le 8$. 
The tiles are fundamental domains for
the {\bf p4} group used to construct them. 
The labels indicate $n$ followed by the tile number for that $n$. Parentheses
indicate that the tiles produce tilings having more symmetries than the {\bf p4} group that generates the tilings.
\label{p4fighi}
}
\end{figure}
\begin{figure}[h]
\centerline{\psfig{file=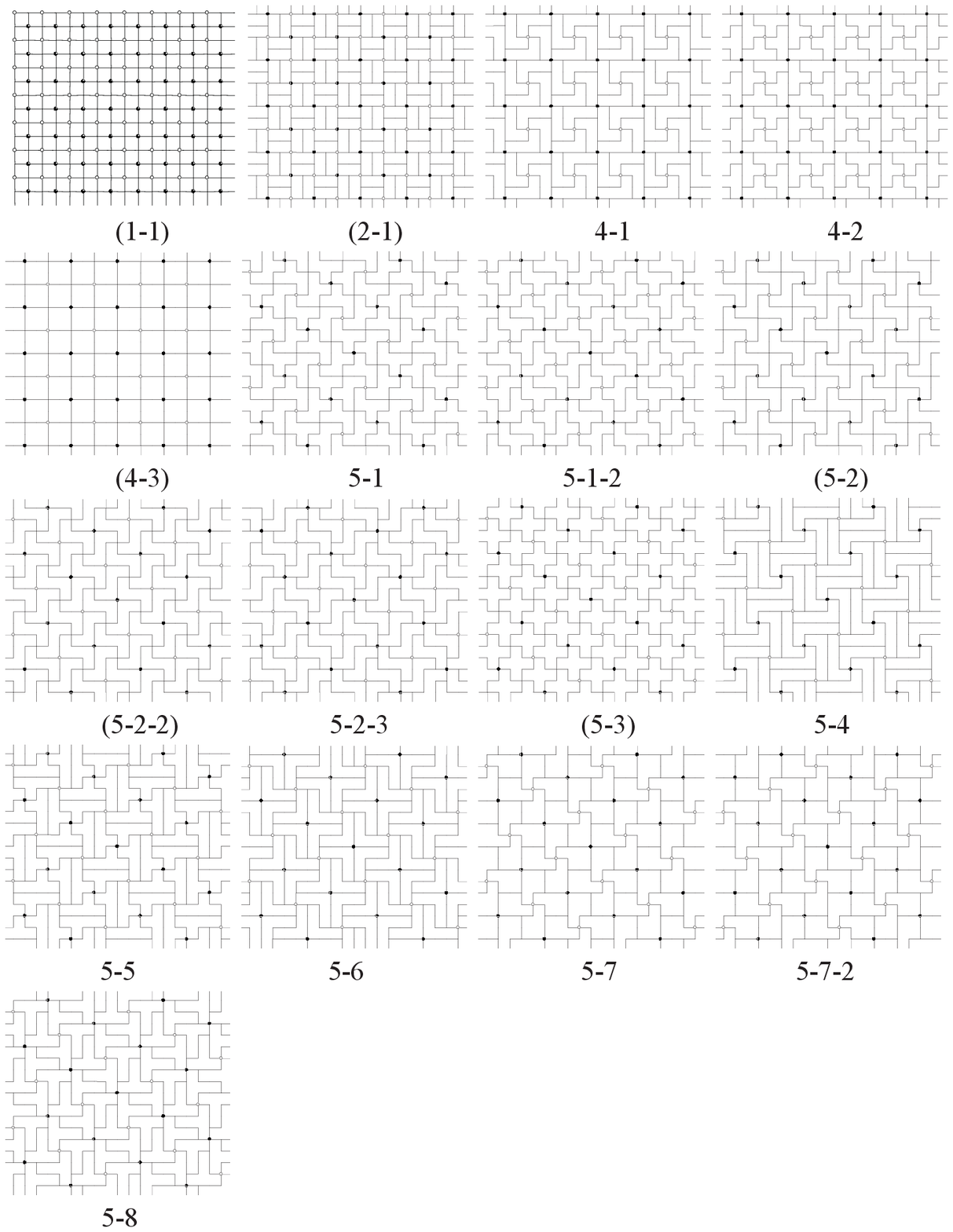,width=0.9\linewidth}}
\caption{
List of {\bf p4} isohedral tilings by $n$-ominoes in Fig.~\ref{p4fighi}, for $n \le 5$, 
generated by a given {\bf p4} group. Labels
correspond to those in Fig.~\ref{p4fighi}. 
The symmetry group of each tiling is the given {\bf p4} group, 
except for tilings whose labels are in parentheses.
\label{p4tiling}
}
\end{figure}

\subsection{Symmetries of tiles}
The list of isohedral tilings corresponding to tiles in Fig.~\ref{p4fighi} includes some tilings having
symmetry group larger than the {\bf p4} group $G$ generated by 4-fold rotations about the black
and white rotation centers. By Theorem~\ref{th1}, in every case where this occurs, the $n$-omino
that generates the tiling must have reflection and/or rotation symmetry.

We outline below how to identify polyominoes whose corresponding isohedral tilings
have symmetry groups larger than the group $G$ that generated the tiling. 
Such tiles in Fig.~\ref{p4fighi} and their tilings in Fig.~\ref{p4tiling} are identified by 
parentheses around their labels.
\begin{itemlist}
\item
Select a polyomino that has rotation and/or reflection symmetry, and examine its
tiling $\cal T$ generated by $G$.
\item
Look at all vertices and centers of unit squares in a polyomino in $\cal T$ except for the
original 4-fold centers we have chosen (black and white centers), and determine whether
or not $\cal T$ is invariant under a 4-fold rotation about such a point. 
If so, then it is a new 4-fold center for $\cal T$, and this symmetry is not in $G$. 
For example, in Fig.~\ref{p4tiling}, tilings (5-2-2)
and (5-2-3), while generated differently, are the same tiling and have 4-fold centers at
every vertex where 4 tiles meet. 
Tiling (5-3) has 4-fold rotation centers at the centers of
the ``cross" tiles and at every vertex where 4 tiles meet. 
These tilings have full symmetry groups of type {\bf p4}, with $G$ as a proper subgroup.
\item
Other new symmetry elements of $\cal T$ can be sought by using chart 2 in\cite{schatt}. 
If the line joining the black and white 4-fold centers we placed is
an axis of mirror reflection for $\cal T$, then $\cal T$ is type {\bf p4m}. 
For example, in Fig.~\ref{p4tiling}, tilings (1-1)
and (4-3) are type {\bf p4m}. 
If this line is not a mirror reflection axis, but the line connecting
two adjacent (nearest) 2-fold centers for $\cal T$ is a mirror reflection axis for $\cal T$, 
then $\cal T$ is type {\bf p4g}. 
For example, in Fig.~\ref{p4tiling}, tiling (2-1) is type {\bf p4g}.
\end{itemlist}

In sections~\ref{p4g} and \ref{p4m} that follow, 
we indicate how to modify our Procedure 1 so as to
produce directly polyominoes that are fundamental domains for {\bf p4g} or {\bf p4m} symmetry
groups.

\section{{\bf p4g}}
\label{p4g}
\subsection{Creating polyominoes as fundamental domains for {\bf p4g} symmetry groups}
\label{p4gcre}
A {\bf p4g} symmetry group contains 4-fold rotations, reflections, and glide-reflections; the
subgroup generated by its 4-fold rotations is type {\bf p4}. 
Thus we begin as in section~\ref{p4},
with a lattice of unit squares. 
Following observation (ii) in section~\ref{intro}, 
we place a 4-fold rotation center (a black circle) at a lattice point and call this the origin. 
Orthogonal unit vectors $\mathbf u$ and $\mathbf v$ are then placed at the origin (see Fig.~\ref{fig:p4glattice}).
\begin{figure}[h]
\centerline{\psfig{file=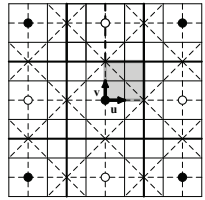,width=2.75cm}}
\caption{
A lattice of unit squares with symmetry elements for a {\bf p4g} group. 
Black and white circles are 4-fold rotation centers, thick solid lines are reflection axes, 
and dashed lines are glide-reflection axes. 2-fold
rotation centers are at the intersections of reflection axes. 
4-fold rotations about the center black circle (the
origin) and reflections about one axis nearest the origin generate the {\bf p4g} group. 
A standard fundamental
domain is shaded. Here $x = 2$; the area of a fundamental domain for the {\bf p4g} group is 4.
\label{fig:p4glattice}
}
\end{figure}

Next, we place the reflection axes that are nearest to the origin; 
these must lie along the edges of unit squares according to observation (i) in section~\ref{intro}. 
These axes are vertical and horizontal lines at $x{\mathbf u}$ and at $x{\mathbf v}$, respectively, 
where $x$ is a positive integer. 
The placement of the origin and choice of $x$ determine the whole {\bf p4g} lattice of rotation
centers, reflection axes and glide-reflection axes since the {\bf p4g} group $G$ is generated 
by 4-fold rotations about the origin and reflections in one of the reflection axes nearest the
origin.

The area $S$ of a fundamental domain for the group $G$ is given by
\begin{equation}
S=x^2
\end{equation}  
where the area of a unit square is 1. Since we want our $n$-omino to be a fundamental
domain, $n = S$. Therefore,
\begin{equation}
n=1, 4, 9, 16, 25, 36, 49, 64, \ldots.
\label{p4g:n}
\end{equation}  

Action by the group $G$ partitions the unit squares into $n$ equivalence classes; as before, we
denote the equivalence class of a unit square $e$ as $C(e)$. 
We construct the set $\mathscr{T}_n$ of $n$-ominoes
that are fundamental domains for $G$ using Procedure 1 in section~\ref{p4}, 
with $n$ chosen from the list in (\ref{p4g:n}) and the additional constraint that unit squares in the set $T$ must
be in the region bounded by the mirror reflection axes nearest to the origin.

Fig.~\ref{p4gfighi} shows the set of inequivalent tiles in $\mathscr{T}_n$ for $n \le 9$. 
We can obtain the associated {\bf p4g} tilings by using the origin as a 4-fold rotation center to fill out the square
bounded by the reflection axes nearest the origin, then reflecting this square in its edges.
Fig.~\ref{p4gtiling} shows the corresponding isohedral tilings produced by polyominoes in Fig.~\ref{p4gfighi}.
\begin{figure}[h]
\centerline{\psfig{file=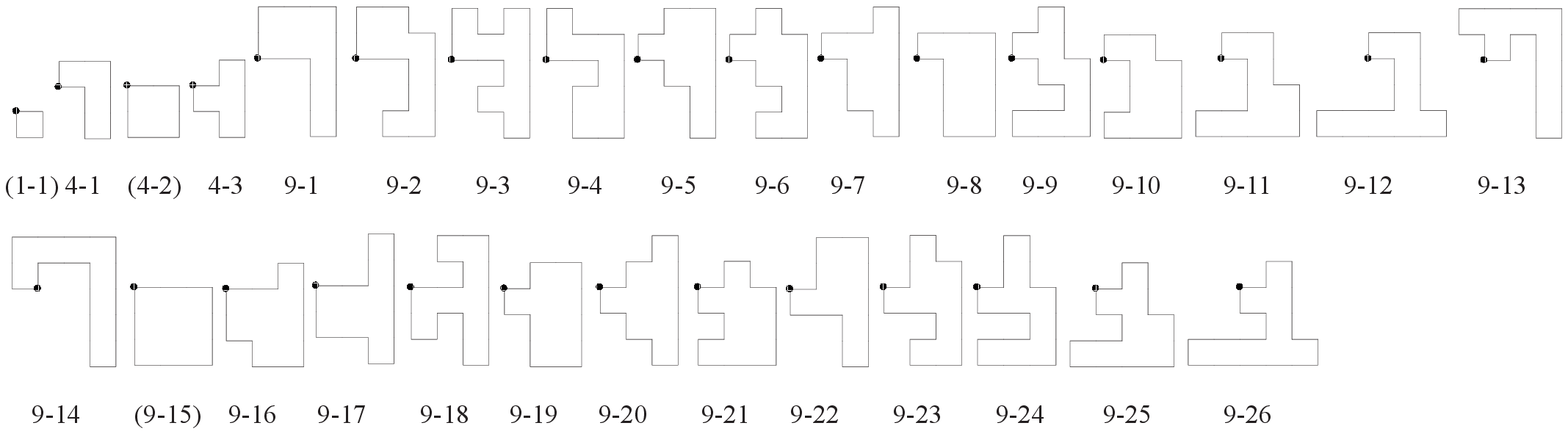,width=\linewidth}}
\vspace*{8pt}
\caption{
List of $n$-ominoes produced as described in section 3.1 for $n \le 9$; 
these are fundamental domains for
the {\bf p4g} group used to construct them. 
The labels indicate $n$ followed by the tile number for that $n$.
Parentheses indicate that the tiles produce tilings having more symmetries than the {\bf p4g} group that
generates the tilings.
\label{p4gfighi}
}
\end{figure}

\begin{figure}
\centerline{\psfig{file=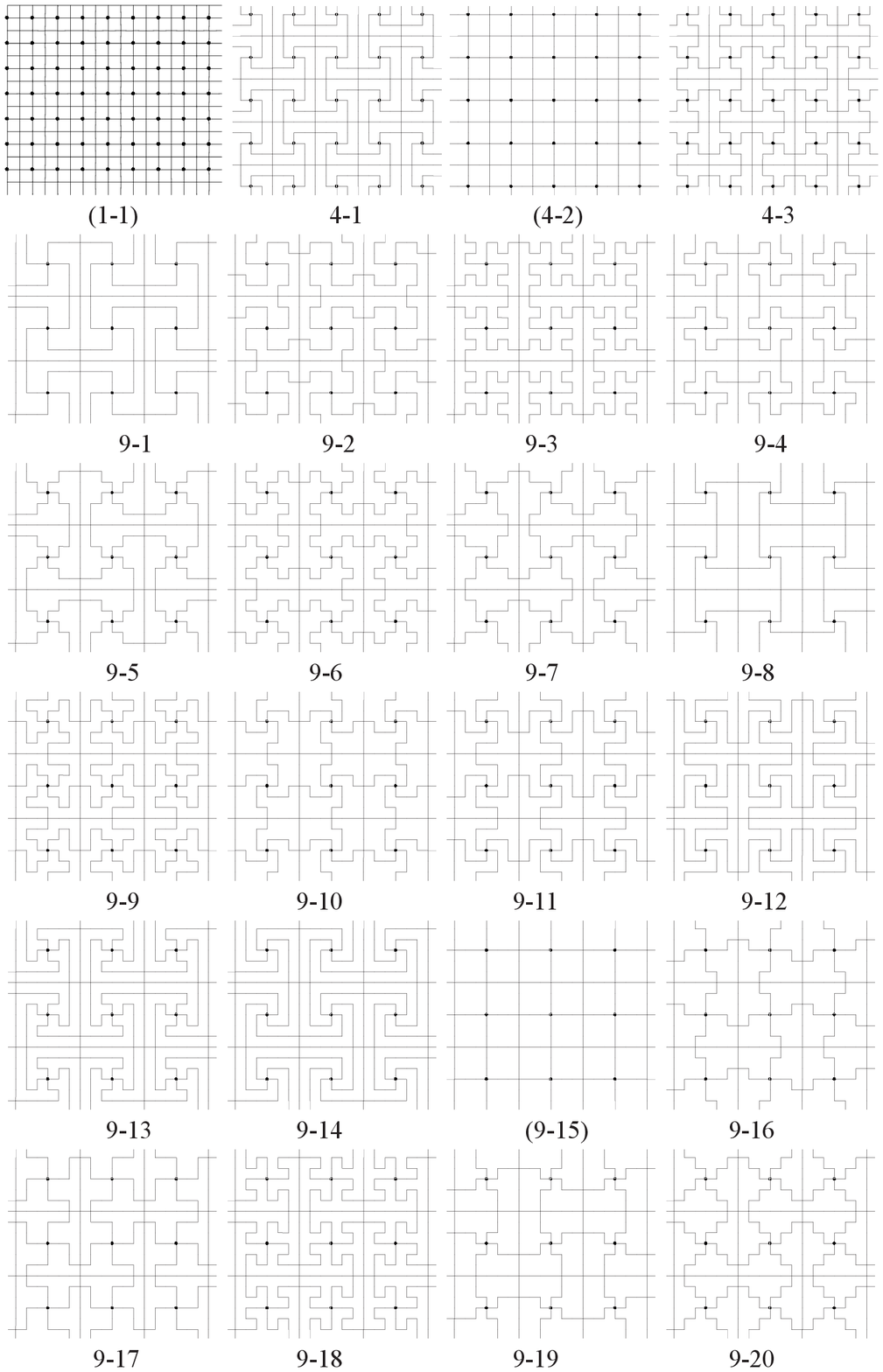,width=0.9\linewidth}}
\vspace*{8pt}
\caption{
List of {\bf p4g} isohedral tilings by $n$-ominoes in Fig.~\ref{p4gfighi}, 
generated by a given {\bf p4g} group. 
Labels correspond to those in Fig.~\ref{p4gfighi}. 
The symmetry group of each tiling is the given {\bf p4g} group, 
except for tilings whose labels are in parentheses.
\label{p4gtiling}
}
\end{figure}
\begin{figure}
\centerline{\psfig{file=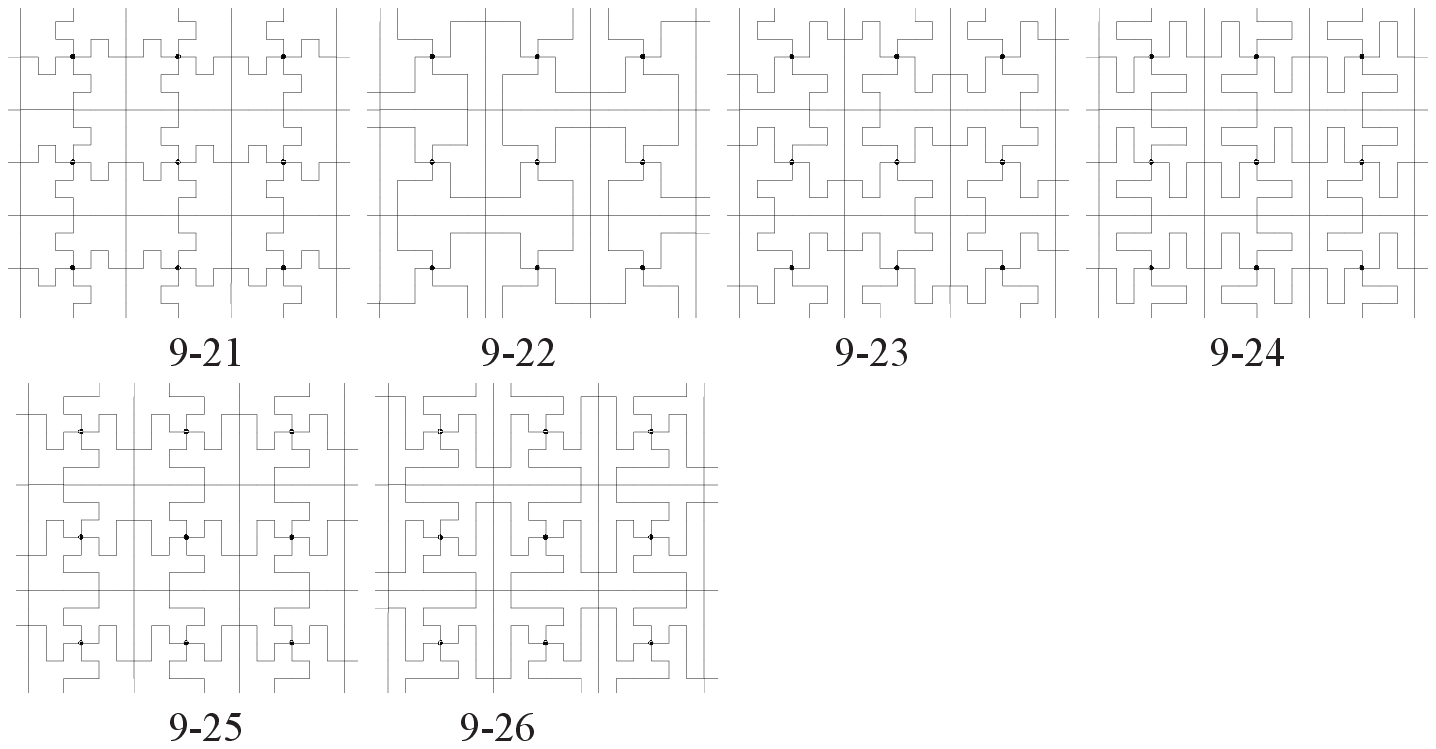,width=0.9\linewidth}}
    \vspace*{8pt}
    \fontsize{8pt}{0pt}\selectfont
    {\it Fig.~\ref{p4gtiling}.} $(${\it Continued}$)$
\end{figure}

\subsection{Symmetries of tiles}
The only tilings in Fig.~\ref{p4gtiling} having symmetry group larger than the {\bf p4g} group $G$ that generated
them are tilings with square polyomino tiles. 
In fact, this is always true.

\begin{theorem}
\label{th2}
Every tiling by a polyomino $T$ produced by our algorithm in section~\ref{p4gcre} has as its
full symmetry group the {\bf p4g} group that generated it and $T$ is a fundamental domain for the tiling,
except in the case when the shape of $T$ is square. 
In that case, the tiling has symmetry group {\bf p4m} and $T$ is not a fundamental domain.
\end{theorem}

\begin{proof}
Let $G'$ be the full symmetry group of a tiling $\cal T$ produced by a {\bf p4g} group $G$ acting on a
polyomino $T$ that is a fundamental domain for $G$, as described in section~\ref{p4gcre}. 
If $G$ is a proper
subgroup of $G'$, then by Theorem~\ref{th1}, 
a rotation or reflection symmetry in $G'$ is also a symmetry of $\cal T$. 
By its construction, one corner of $T$ is a 4-fold center for $G$ at the origin; 
without loss of generality, we may assume that $T$ contains the unit square shown in black in Fig.~\ref{X}. 
\begin{figure}[h]
\centerline{
\psfig{file=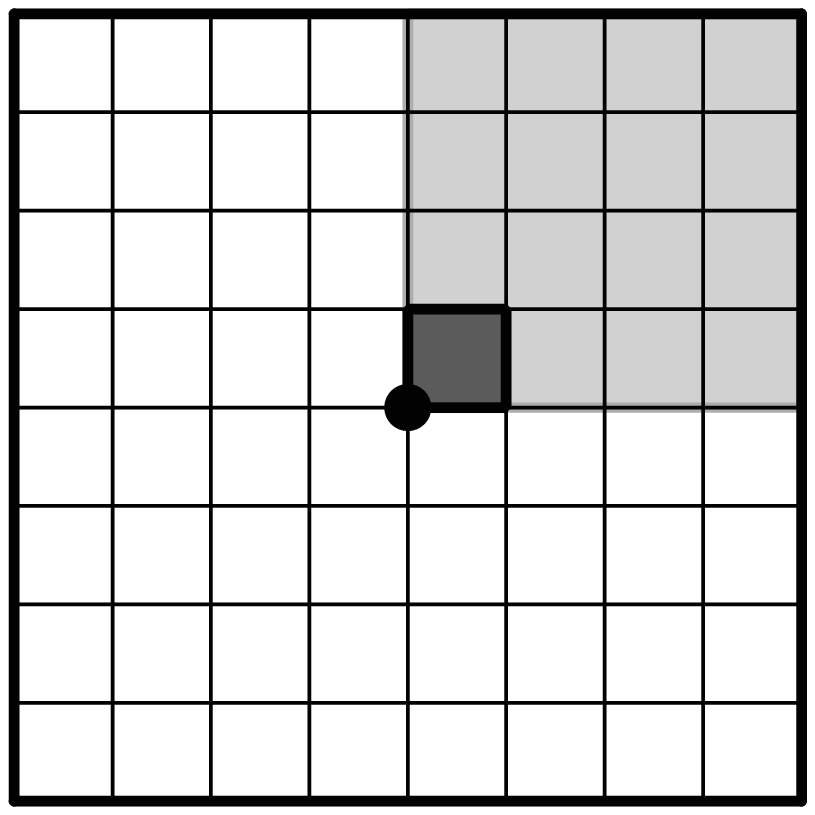,width=1.7cm}\psfig{file=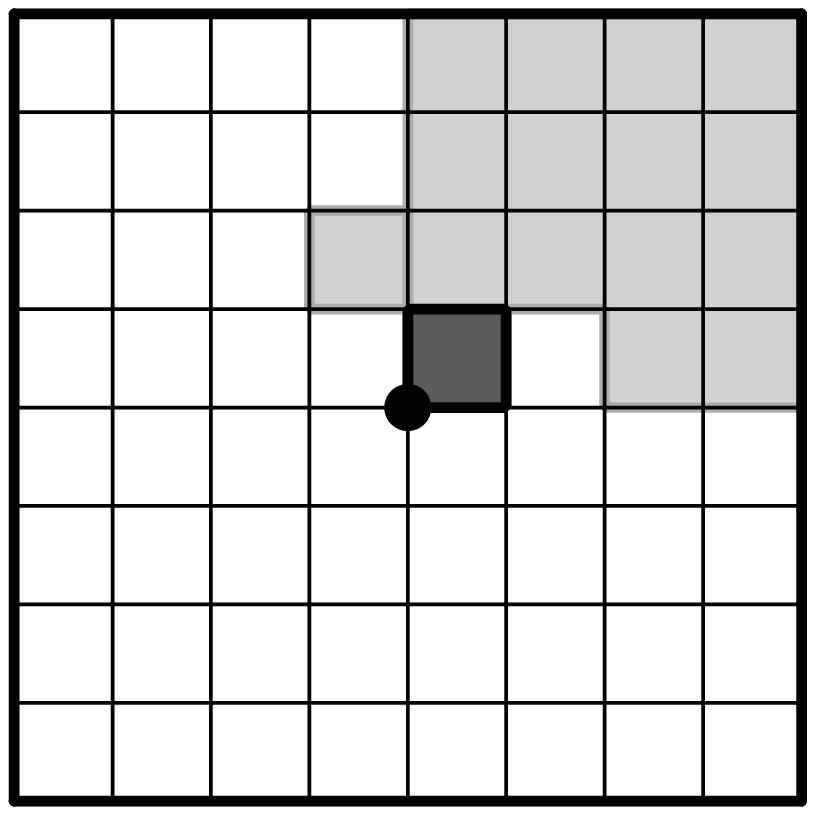,width=1.7cm}\psfig{file=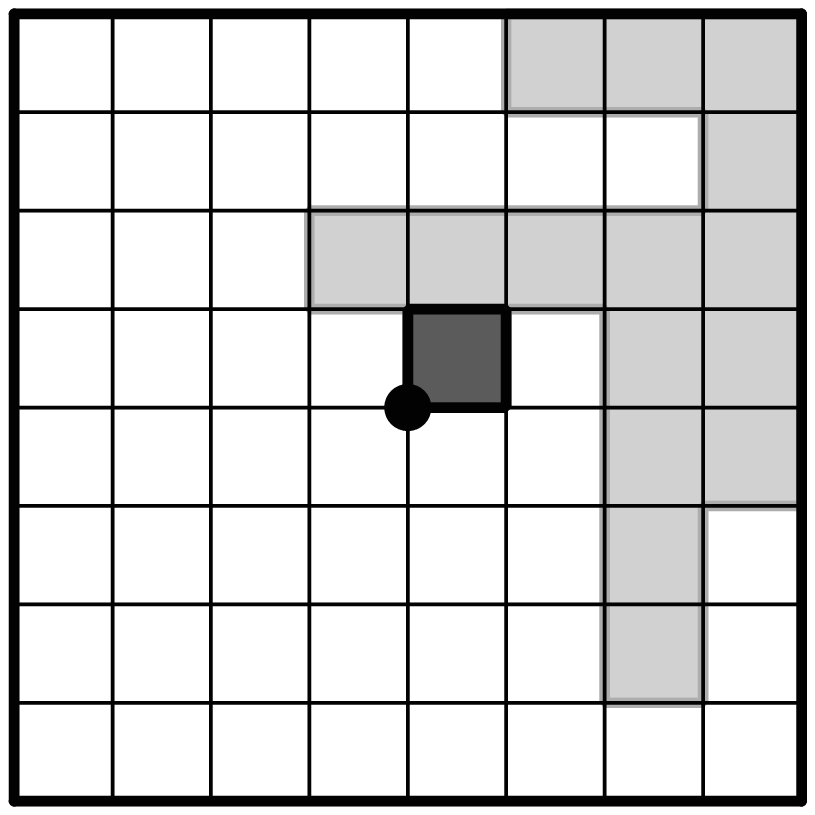,width=1.7cm}\psfig{file=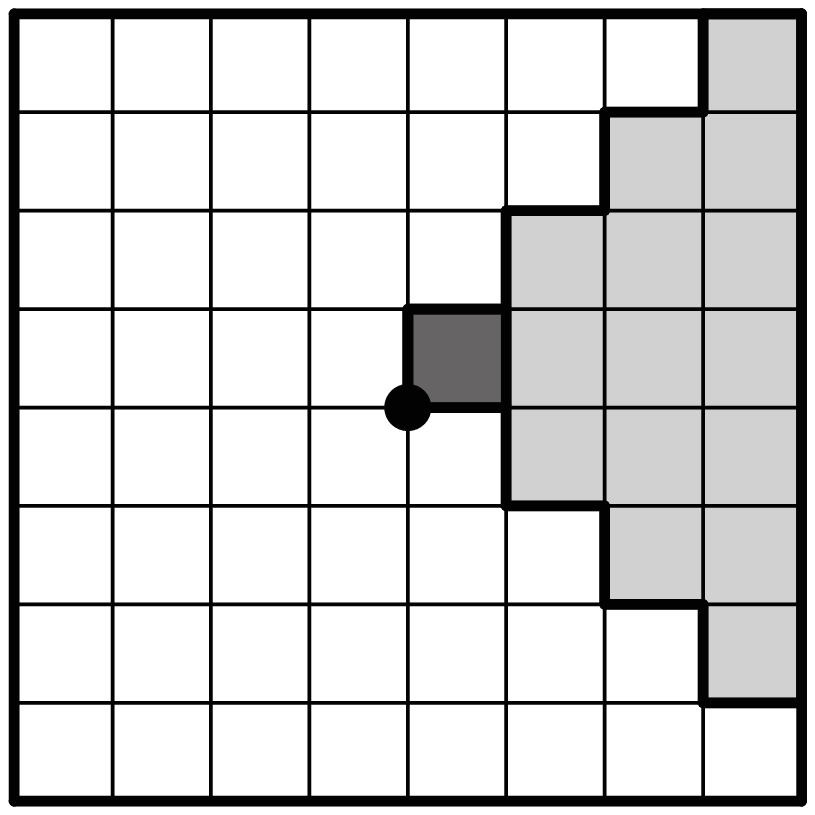,width=1.7cm}
}
\centerline{(a) \hspace{1.0cm} (b) \hspace{1.0cm} (c) \hspace{1.0cm} (d)}
\vspace*{8pt}
\caption{
The derivation of polyomino fundamental domains $T$ for a {\bf p4g} group by removing unit squares from and
adding equivalent unit squares to the original square polyomino fundamental domain in (a). Here $x = 4$.
\label{X}
}
\end{figure}
$T$ can also be obtained beginning with an $x \times x$ square polyomino tile as in Fig.~\ref{X}(a), 
then removing unit squares and adding other unit squares equivalent (by a rotation of $90^\circ$ or $270^\circ$ about 
the origin) to
these, always keeping the new tile homeomorphic to a disc. 
Figures~\ref{X}(b), (c), and (d) illustrate this. 
From this, we see that one edge of $T$ that abuts a reflection axis in $G$ has at least $x$ unit
squares in a straight row. Since the two edges of $T$ that join the origin to each of the reflection
axes at $x {\mathbf u}$ and $x {\mathbf v}$ are related by a $90^\circ$ rotation and 
are not straight unless $T$ is a square, 
the only possible $T$ that can have rotation symmetry is the square polyomino in Fig.~\ref{X}(a).

Suppose that $T$ is not square but has reflection symmetry. 
If the two edges of $T$ that abut the
reflection axes at $x{\mathbf u}$ and $x {\mathbf v}$ are the same length, 
that reflection must map those edges to each other. 
But that diagonal reflection will not leave $T$ invariant (see Fig.~\ref{X}(b)). 
If the two edges of $T$ that abut the reflection axes at $x {\mathbf u}$ and $x {\mathbf v}$ have 
different lengths, a reflection symmetry of $T$
must leave fixed the longer edge, so must have its axis a horizontal or vertical line. 
But this reflection cannot be a symmetry in $G'$. 
For example, if $T'$ is the image of $T$ rotated by $90^\circ$ about
the origin, this reflection will not map $T'$ onto another tile in $\cal T$ (see Fig.~\ref{X}(d)).
\end{proof}

\section{{\bf p4m}}
\label{p4m}
The symmetry elements of a {\bf p4m} group are shown in Fig.~\ref{fig:p4m}.
\begin{figure}[h]
\centerline{
\psfig{file=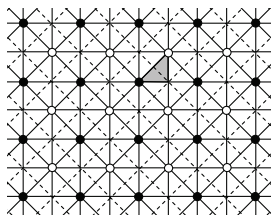,width=3.5cm}}
\vspace*{8pt}
\caption{
The symmetry elements of a {\bf p4m} group. 
Black and white circles are 4-fold centers, solid lines are reflection
axes and dashed lines are glide-reflection axes. 
Points where glide-reflection axes intersect are 2-fold rotation
centers. The shaded area is a fundamental domain for the {\bf p4m} group.
\label{fig:p4m}
}
\end{figure}
Black and white circles are 4-fold centers, solid lines are mirror reflection axes, 
and dashed lines are glide-reflection axes. Points
where glide-reflection axes intersect are 2-fold centers. The shaded region is a fundamental
domain for the {\bf p4m} group. A polyomino that is a fundamental domain for the tiling must have
its edges on the reflection axes, by observation (i) in section~\ref{intro}. 
But this is clearly impossible.

\begin{theorem}
There are no {\bf p4m} isohedral tilings having polyominoes as fundamental domains.
\end{theorem}

\section{{\bf p3}}
\label{p3}
\subsection{Creating polyiamonds as fundamental domains for {\bf p3} symmetry groups}
\label{p3:create}
To build our $n$-iamond tiles that will be fundamental domains for a {\bf p3} isohedral tiling, we begin
with a lattice of unit equilateral triangles. 
By observation (iii) in section~\ref{intro}, 
the 3-fold rotation centers for a {\bf p3} symmetry group are located at the lattice points. 
So first we place a 3-fold
rotation center, a black circle, at a lattice point and call this the origin, then place vectors 
$\mathbf u$ and $\mathbf v$
at the origin along edges of a unit triangle. 
Next we place a second 3-fold rotation center, a
white circle, at $x{\mathbf u} + y{\mathbf v}$, where $x$ and $y$ are nonnegative integers, 
not both zero. See Fig.~\ref{fig:p3}(a).
\begin{figure}[h]
\centerline{
\psfig{file=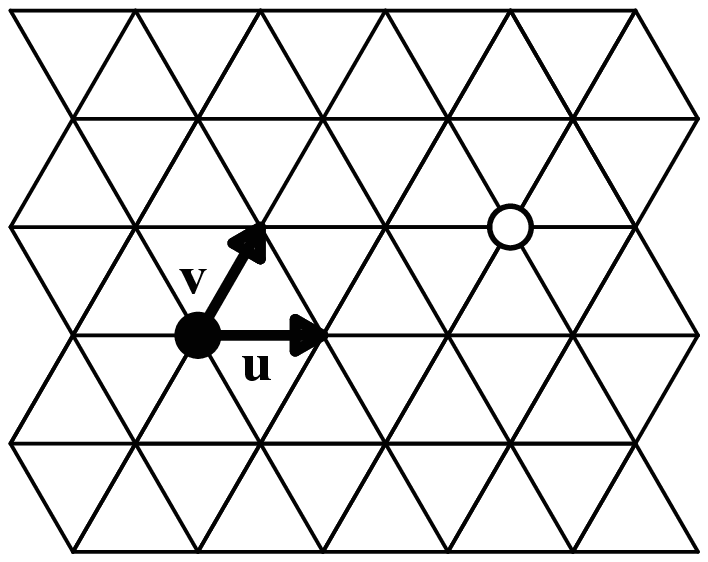,width=2.7cm}\psfig{file=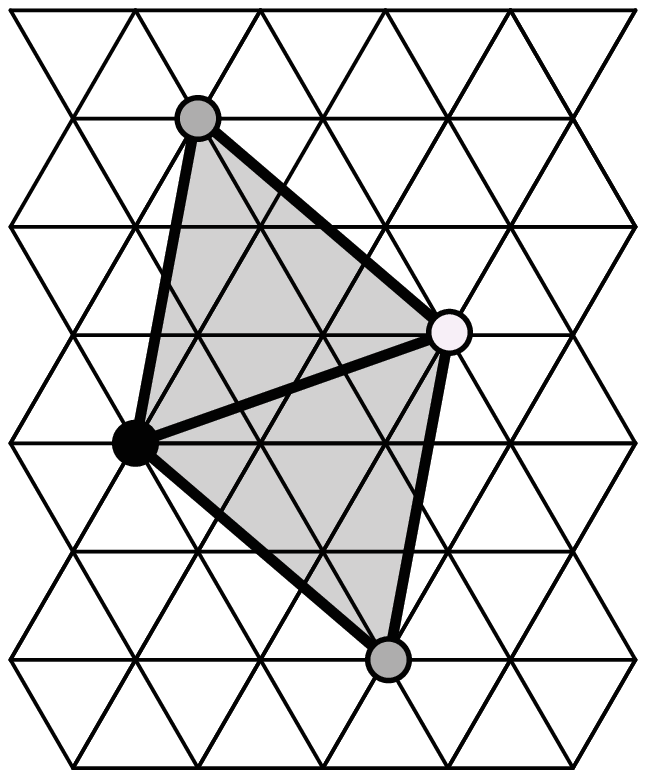,width=2.52cm}}
\centerline{(a) \hspace{2.1cm} (b)}
\vspace*{8pt}
\caption{
(a) A lattice of unit triangles with two 3-fold rotation centers that generate a {\bf p3} group. 
Here $x = 2$, $y = 1$. 
(b) The shaded region is a standard fundamental domain for the {\bf p3} group generated by the black and white
3-fold centers; in this example, the area of a fundamental domain is 14 triangular units.
\label{fig:p3}
}
\end{figure}
These two choices of 3-fold rotation centers determine the whole {\bf p3} lattice of rotation centers; 
3-fold rotations about the black and white centers generate the whole {\bf p3} symmetry group.

A standard fundamental domain for the {\bf p3} group $G$ generated by the black and white 3-fold
centers is a $120^\circ$ rhombus having as short diagonal the segment that joins the black and white
circles, as shown in Fig.~\ref{fig:p3}(b). The other two vertices of this rhombus (gray in Fig.~\ref{fig:p3}(b)) are also
3-fold centers for $G$. The area $S$ of a fundamental domain for $G$ is given by
\begin{equation}
S=2 (x^2+y^2+xy)   
\end{equation}
taking the area of a unit triangle as 1. 
Since we want our $n$-iamond to be a fundamental
domain, $n = S$. Therefore,
\begin{equation}
n=2, 6, 8, 14, 18, 24, 26, 32, 38, 42, 50, 54, 56, 62, 72, 74, 78, 86, 96, 98, 98, \ldots,
\label{p3:n}
\end{equation}
where two pairs $(x,y)$, namely, $(5,3)$ and $(7,0)$ correspond to $n = 98$.

All unit triangles are classified into n equivalence classes by the action of $G$, and we denote
the equivalence class of a unit triangle $e$ as $C(e)$.

We construct a set $\mathscr{T}_n$ of $n$-iamonds
that are fundamental domains for $G$ by modifying
Procedure 1 in section~\ref{p4}, 
choosing $n$ from the list in (\ref{p3:n}), and replacing definitions in section~\ref{p4} with these:

\begin{enumerate}
\item 
$T$ is a set of unit triangles; $B(T)$ is a set of unit triangles that are edge-adjacent to the
triangles in $T$; $\mathscr{T}_n$ is a set of $n$-iamonds.
\item
When $T$ is the empty set, we define $B(\emptyset)$ as the set of six triangles around the origin in
Fig.~\ref{fig:p3}(a).
\item
We define $E(T)$, the boolean function of $T$, which is true if $\# T = n$ and the white circle is on
the boundary of $T$ . Otherwise $E(T)$ is false.
\item
$B'(T) = \{e | e \in B(T),C(e) \ne C(f), \mbox{for all} f \in T \}$. 
This is the set of all unit triangles that are edge-adjacent to those in $T$, 
but not equivalent to any unit triangles in $T$.
\end{enumerate}
The procedure creates a sequence of pairs of sets $(T, U_T)$, 
in which $U_T$ is the set of unit triangles that can be added to $T$ to
create the next set $T$ in the sequence.

Fig.~\ref{p3fighi} shows the set of inequivalent $n$-iamonds in $\mathscr{T}_n$
for $n \le 8$. From each of these we can
obtain the associated {\bf p3} tiling by using the black and white circles as 3-fold rotation centers.
Fig.~\ref{p3tiling} shows the corresponding isohedral tilings produced by $n$-iamonds in Fig.~\ref{p3fighi}.
\begin{figure}[h]
\centerline{
\psfig{file=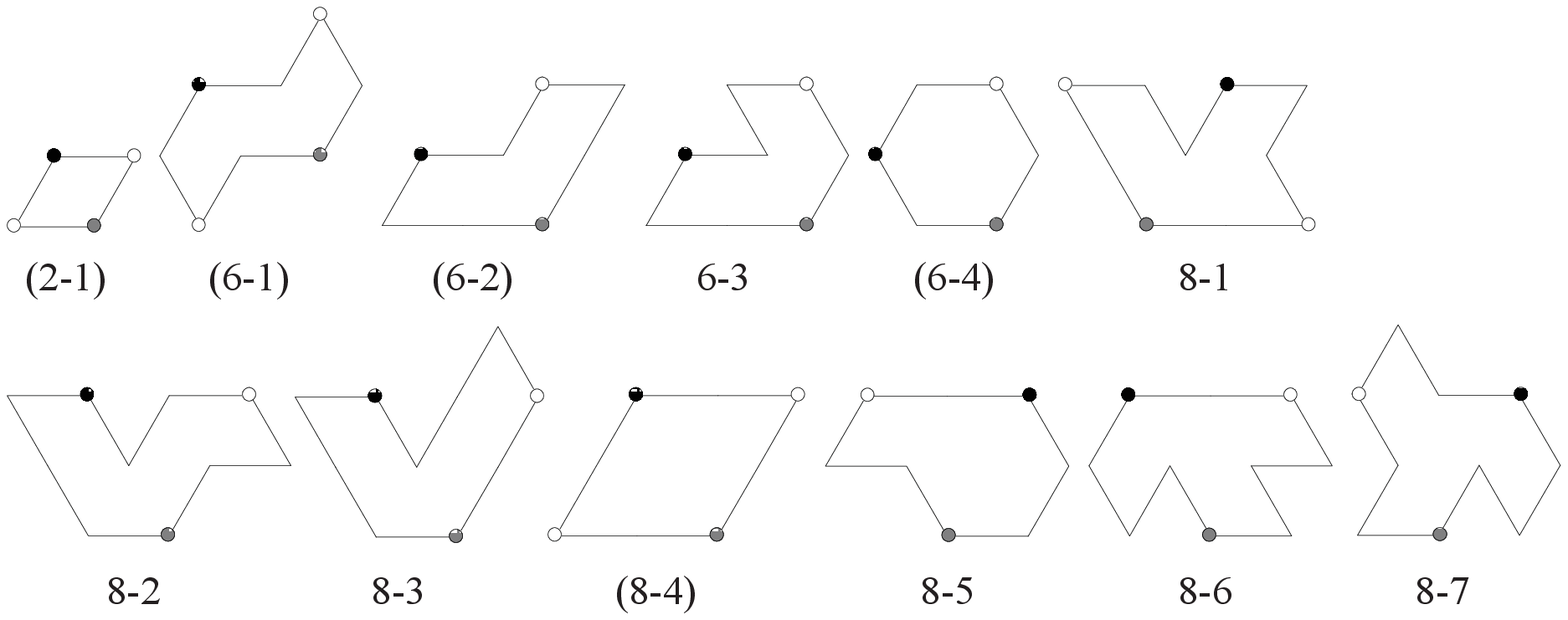,width=0.9\linewidth}
}
\vspace*{8pt}
\caption{
List of $n$-iamonds produced as described in section~\ref{p3:create} for $n \le 8$. 
The tiles are fundamental domains for the
{\bf p3} group used to construct them. 
The labels indicate $n$ followed by the tile number for that $n$. 
Parentheses indicate that the tiles produce tilings having more symmetries 
than the {\bf p3} group that generates the tilings.
\label{p3fighi}
}
\end{figure}
\begin{figure}[h]
\centerline{
\psfig{file=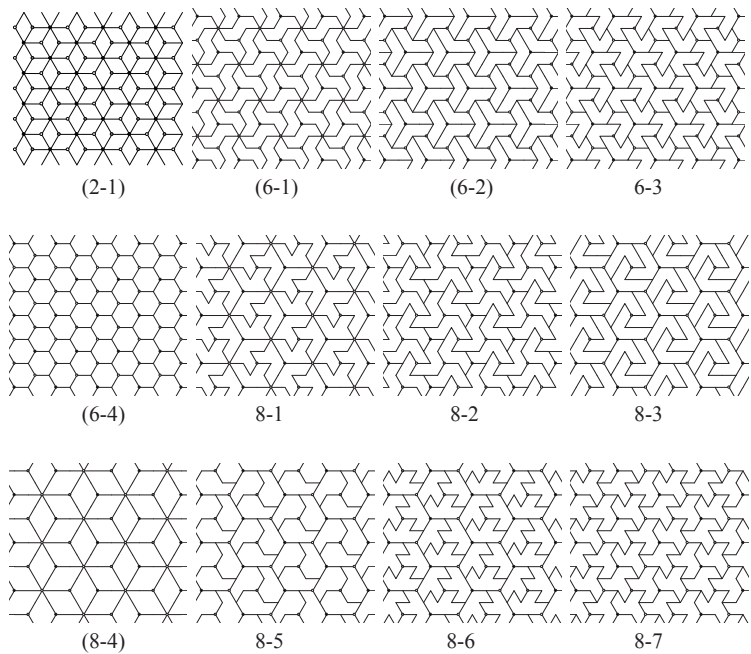,width=0.8\linewidth}}
\vspace*{8pt}
\caption{
List of {\bf p3} isohedral tilings by $n$-iamonds in Fig.~\ref{p3fighi}, 
generated by a given {\bf p3} group. 
Labels correspond to those in Fig.~\ref{p3fighi}. 
The symmetry group of each tiling is the {\bf p3} group that generated it, 
except for tilings whose labels are in parentheses.
\label{p3tiling}
}
\end{figure}

\subsection{Symmetries of tiles}
\label{p3:sym}
The list of isohedral tilings in Fig.~\ref{p3tiling} includes some tilings having symmetry group larger than
the {\bf p3} group $G$ generated by 3-fold rotations about the black and white rotation centers. 
By Theorem~\ref{th1}, when this occurs, the $n$-iamond that generates the tiling must have reflection and/or
3-fold rotation symmetry. 
The full symmetry group of the tiling could be of any of these types:
{\bf p3}, {\bf p3m1}, {\bf p31m}, {\bf p6}, and {\bf p6m}. 
The tilings in Fig.~\ref{p3tiling} that have additional symmetries, and the
$n$-iamonds in Fig.~\ref{p3fighi} that generate them are indicated by parentheses in their labels. 
We outline below how to identify these polyiamonds.

\begin{itemlist}
\item
Select a polyiamond $T$ that has rotation and/or reflection symmetry, and examine its tiling
$\cal T$ generated by $G$.
\item
Look at all vertices and centers of unit triangles in $T$ except for those centers of rotation in
$G$ (black, white and gray centers), and determine whether or not any can be new 3-fold centers
for $\cal T$. 
If a new 3-fold center is found, then the full symmetry group $G'$ of $\cal T$ contains $G$ as a
proper subgroup and $T$ is not a fundamental domain for $\cal T$. 
If a new 3-fold center is a 6-fold
center, $G'$ is type {\bf p6m} if there is a reflection axis joining it to any nearest 3-fold center;
otherwise, $G'$ is type {\bf p6}. 
Tiling (6-4) in Fig.~\ref{p3tiling} is type {\bf p6m}, and has a 6-fold center at the center
of the hexagonal 6-iamond. 
If there are new 3-fold centers but none are 6-fold centers, then $G'$ is type {\bf p3}.
\item
If the tiling contains only 3-fold centers for $G$, 
we seek new symmetry elements by using chart 2 in\cite{schatt}. 
\begin{enumerate}
\item
If some 3-fold center for $G$ is a 6-fold center, 
the tiling $\cal T$ will have {\bf p6m} symmetry if the
line joining a 6-fold center to a nearest 3-fold center is a reflection axis for $\cal T$, 
otherwise it will have {\bf p6} symmetry. 
In Fig.~\ref{p3tiling}, tilings (2-1) and (8-4) have {\bf p6m} symmetry, 
and tiling (6-1) has {\bf p6} symmetry.
\item
Otherwise, if the line connecting two adjacent different 3-fold centers 
(say, black to white) is a reflection axis, $\cal T$ has {\bf p3m1} symmetry.
\item
Otherwise, if the line connecting adjacent 3-fold centers of the same kind 
(black-black, white-white, or gray-gray) is a reflection axis, 
$\cal T$ has {\bf p31m} symmetry. 
In Fig.~\ref{p3tiling}, tiling (6-2) has {\bf p31m} symmetry.
\end{enumerate}
\end{itemlist}

\section{{\bf p31m}}
\label{p31m}
\subsection{Creating polyiamonds as fundamental domains for {\bf p31m} symmetry groups}
\label{p31m:create}
We begin with a lattice of unit triangles as in section~\ref{p3:create}, 
since a {\bf p31m} symmetry group contains 3-fold rotations, reflections, and glide-reflections; 
the subgroup generated by its 3-fold rotations is type {\bf p3}. 
\begin{figure}[h]
\centerline{
\psfig{file=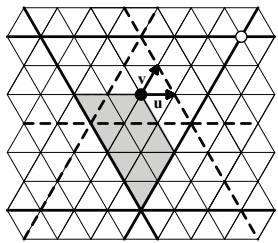,width=0.3\linewidth}}
\vspace*{8pt}
\caption{
A lattice of unit triangles with a 3-fold rotation center at the origin (black circle), reflection axes 
(thick solid lines) and glide-reflection axes (dashed lines) that are symmetries in a {\bf p31m} group. 
The 3-fold rotation about the
black circle and reflections about one axis generate the {\bf p31m} group, 
which also has 3-fold centers where reflection
axes intersect (one of these is shown as a white circle). 
A typical fundamental domain for the group is shaded. 
Here $x = 2$; the area of the fundamental domain is 12 triangular units.
\label{fig:p31m}
}
\end{figure}
To build an $n$-iamond tile that is a fundamental domain for a {\bf p31m} isohedral tiling,
we first place a 3-fold rotation center, a black circle, at a lattice point and call this the origin.
Then we place vectors $\mathbf u$ and $\mathbf v$ at the origin along edges of a unit triangle (see Fig.~\ref{fig:p31m}). 

Next, we place reflection axes that are nearest the origin; 
these must lie along the edges of
unit triangles according to observation (i) in section \ref{intro}. 
These three axes intersect at the points
$x ({\mathbf u}+{\mathbf v})$, $x(-2{\mathbf u} +{\mathbf v})$, and $x({\mathbf u} -2{\mathbf v})$, 
where $x$ is a positive integer; 
their intersections are equivalent 3-fold rotation centers 
(one of these is shown as a white circle in Fig.~\ref{fig:p31m}). 
The placement of the origin and choice of $x$ determine the whole {\bf p31m} lattice of 
rotation centers, reflection axes and glide-reflection axes. 
The {\bf p31m} group $G$ is generated by 3-fold rotations about the origin and
reflections in one of the three reflection axes.

The area $S$ of a fundamental domain for the tiling is one-third the area enclosed by the
three reflection axes closest to the origin, which is
\begin{equation}
S=3 x^2
\end{equation}
where the area of a unit triangle is 1. 
Since we want our $n$-iamond to be a fundamental domain, $n = S$. 
Therefore,
\begin{equation}
n=3, 12, 27, 48, 75, 108, 147, 192, 243, \ldots.
\label{p31m:n}
\end{equation}  

The action of $G$ classifies all unit triangles into $n$ equivalence classes 
and we denote the equivalence class of a unit triangle $e$ as $C(e)$. 
We construct a set $\mathscr{T}_n$ of $n$-iamonds
that are
fundamental domains for $G$ as described in section~\ref{p3:create}, 
choosing $n$ from the list in (8) and
having the additional constraint that the unit triangles in the set $T$ 
must be in the region bounded by the mirror reflection axes in Fig.~\ref{fig:p31m}.

Fig.~\ref{p31mfighi} shows the set of inequivalent $n$-iamonds in $\mathscr{T}_n$  for $n \le 12$.  
Fig.~\ref{p31mtiling} shows the corresponding {\bf p31m} tilings obtained 
by performing a 3-fold rotation about the origin (black circle) to fill out 
the triangle bounded by reflection axes (Fig.~\ref{fig:p31m}), 
then reflecting this triangle in its edges.
\begin{figure}[h]
\centerline{
\psfig{file=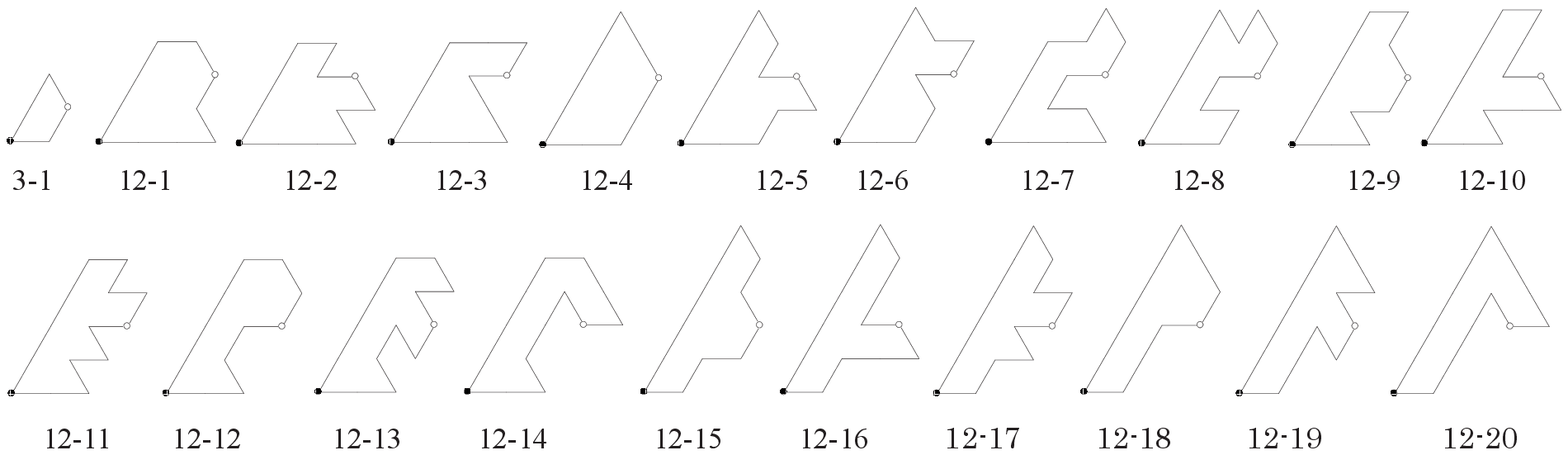,width=1.0\linewidth}}
\vspace*{8pt}
\caption{
List of $n$-iamonds produced as described in section \ref{p31m:create} for $n \le 12$. 
The tiles are fundamental domains for the {\bf p31m} group used to construct them. 
The labels indicate $n$ followed by the tile number for that $n$. 
\label{p31mfighi}
}
\end{figure}
\begin{figure}
\centerline{
\psfig{file=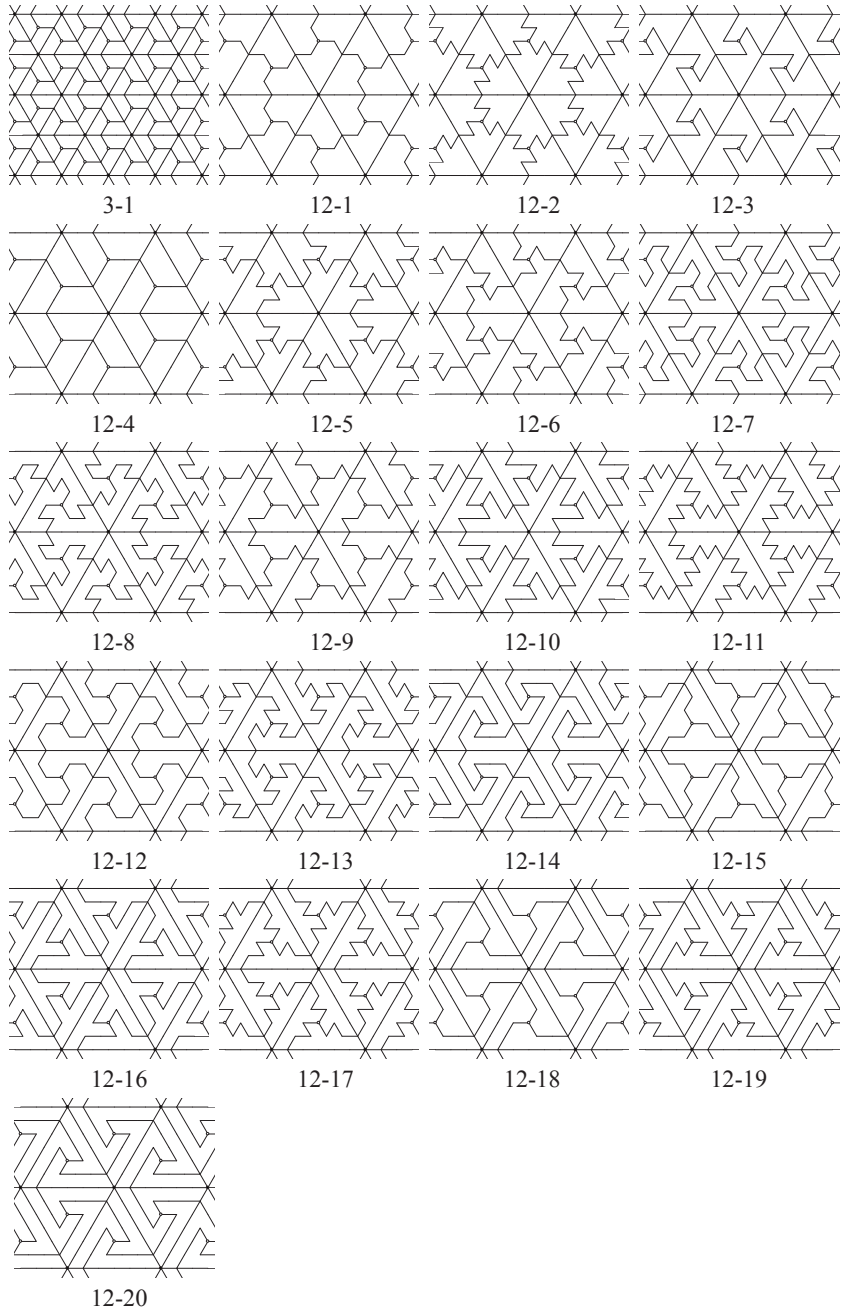,width=0.9\linewidth}}
\vspace*{8pt}
\caption{
List of {\bf p31m} isohedral tilings by $n$-iamonds in Fig.~\ref{p31mfighi}, 
generated by a given {\bf p31m} group, for $n \le 12$.
Labels correspond to those in Fig.~\ref{p31mfighi}. 
Every $n$-iamond in Fig.~\ref{p31mfighi} is a fundamental domain for its tiling in Fig.~\ref{p31mtiling}.
\label{p31mtiling}
    }
\end{figure}

\subsection{Symmetries of tiles}

\begin{theorem}
Every isohedral tiling by a polyiamond $T$ produced by our algorithm in section \ref{p31m:create}
has as its symmetry group the {\bf p31m} group that generated it and $T$ is a fundamental domain for
the tiling.
\end{theorem}

\begin{proof}
Our proof is analogous to that for Theorem~\ref{th2}. 
Let $G'$ be the full symmetry group of a
tiling $\cal T$ produced by a {\bf p31m} group $G$ acting on a polyiamond $T$ that is a fundamental domain for $G$, 
as described in section \ref{p31m:create}. 
If $G$ is a proper subgroup of $G'$, then by Theorem~\ref{th1}, 
a reflection symmetry or 3-fold rotation symmetry in $G'$ is also a symmetry of $T$. 
$T$ can be obtained from the shaded fundamental domain shown in Fig.~\ref{fig:p31m} 
by removing unit triangles and adding unit triangles equivalent to these 
(by a $120^\circ$ rotation about the origin), 
always keeping the new tile homeomorphic to a disc. 
This process shows that one edge of $T$ that lies on a reflection axis will
have length at least $x+1$, and at most one other edge of $T$ (lying on an adjacent reflection axis of $G$) 
will have this same length (see Fig.~\ref{p31mfighi}). 
From this, it follows that $T$ cannot have 3-fold rotation symmetry.

Suppose that $T$ has reflection symmetry. 
If the two straight edges of $T$ that lie on reflection
axes for $G$ have the same length, then the reflection that leaves $T$ fixed must map these edges to each other. 
But this reflection cannot leave $T$ invariant, since the other edges of $T$ are related by
a 3-fold rotation about the origin (see, for example, tiles 12-1, 12-2, 12-3 in Fig.~\ref{p31mfighi}). 
If the two straight edges of $T$ that lie on reflection axes for $G$ have different lengths, 
then a reflection that fixes $T$ must leave the longer edge fixed, 
and so is perpendicular to that edge. 
This can only happen if $T$ is the rhombus fundamental domain in Fig.~\ref{fig:p31m}, 
and in this case, the reflection symmetry for $T$ is not a symmetry for the tiling $\cal T$ 
(see tiling 12-4 in Fig.~\ref{p31mtiling}).
\end{proof}

\section{{\bf p3m1}}

The lattice of reflections and 3-fold rotation centers of a {\bf p3m1} symmetry group is shown 
in Fig.~\ref{fig:p3m1}; 
black, white, and grey circles denote three inequivalent 3-fold rotation centers. 
Here, unlike the {\bf p31m} case, 
all 3-fold centers lie on reflection axes. 
The shaded region bounded by reflection axes is a fundamental domain for the {\bf p3m1} group 
that generates the tiling by reflections those axes. 
\begin{figure}[h]
\centerline{
\psfig{file=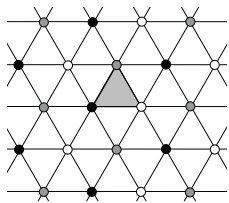,width=0.23\linewidth}}
\vspace*{8pt}
\caption{
The lattice of reflection axes and 3-fold rotation centers of a {\bf p3m1} group 
(glide-reflection axes are not shown). 
Three inequivalent 3-fold centers are black, grey, and white circles. 
Lines are reflection axes. 
The shaded region is a fundamental domain for the {\bf p3m1} group generated by reflections 
in the axes surrounding the region.
\label{fig:p3m1}
    }
\end{figure}
By observation (i) in section \ref{intro}, 
this is the only polyiamond tile possible having the area of a fundamental domain. 
But the full symmetry group of this tiling is type {\bf p6m}, 
which has a fundamental domain with area 1/6 that of the shaded tile. Thus,

\begin{theorem}
There are no {\bf p3m1} isohedral tilings having polyiamonds as fundamental domains.
\end{theorem}

We note that if the shaded region in Fig.~\ref{fig:p3m1} is decorated with an asymmetric motif 
then there are $n$-iamonds (where $n = k^2$, $k$ a positive integer) having the shape of an equilateral triangle, 
for which the decorated triangle is a fundamental domain for a {\bf p3m1} isohedral tiling.

\section{{\bf p6}}
\label{p6}
\subsection{Creating polyiamonds as fundamental domains for {\bf p6} symmetry groups}
\label{p6:create}
A {\bf p6} symmetry group contains 6-fold, 3-fold, and 2-fold rotations; 
the subgroup generated by its 3-fold rotations is type {\bf p3}. 
Thus we begin as in section \ref{p3:create}, with a lattice of equilateral triangles,
to build $n$-iamond tiles that are fundamental domains for a {\bf p6} group. 
By observation (iii) in section \ref{intro}, 
the 6-fold and 3-fold rotation centers for a {\bf p6} symmetry group are located at lattice points. 
So first we place a 6-fold rotation center, a black circle, at a lattice point and call this the origin, then
place vectors $\mathbf u$ and $\mathbf v$ at the origin, along edges of a unit triangle. 
Next we place a 3-fold rotation center, 
a white circle, at $x{\mathbf u} + y{\mathbf v}$, 
where $x$ and $y$ are nonnegative integers, 
not both $0$. 
\begin{figure}[h]
\centerline{
\psfig{file=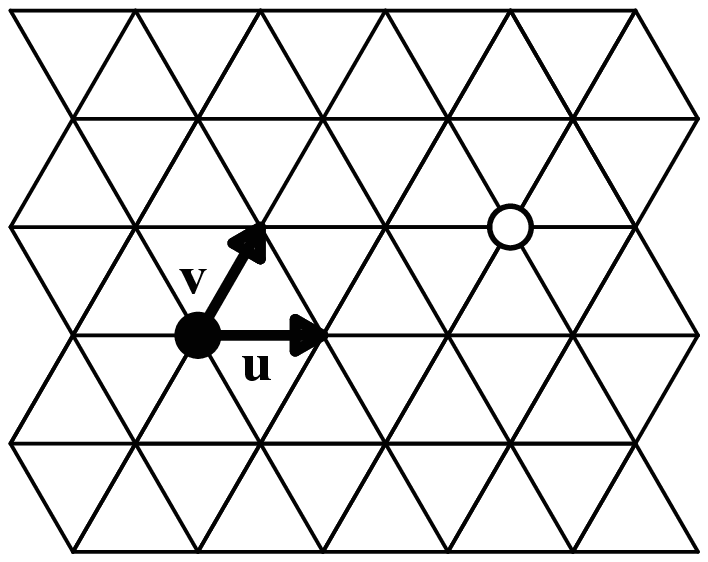,width=2.88cm}
\hspace{1em}
\psfig{file=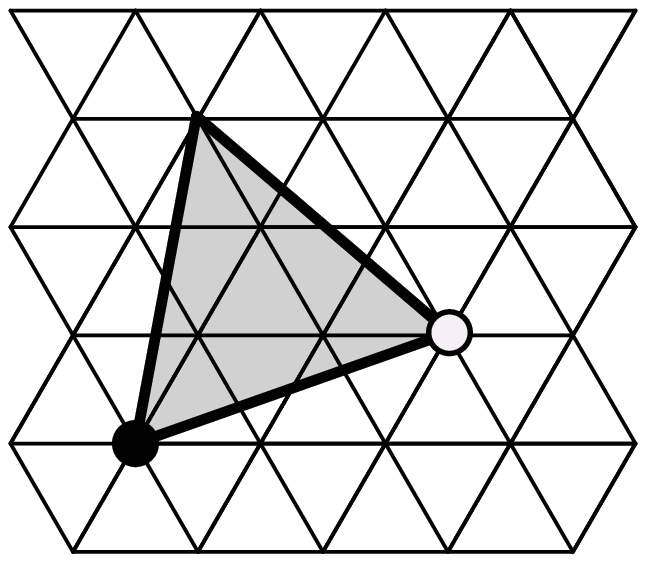,width=2.7cm}
}
\centerline{(a) \hspace{3cm} (b)}
\vspace*{8pt}
\caption{
(a) A lattice of unit triangles with a black 6-fold rotation center and white 3-fold rotation center that
generate a {\bf p6} group. 
Here $x = 2$, $y = 1$. (b). A standard fundamental domain for the {\bf p6} group generated by the
black and white rotation centers; 
in this example, the area of the fundamental domain is $7$ triangular units.
\label{p6lattice}
}
\end{figure}
See Fig.~\ref{p6lattice}. 
These choices of 6-fold and 3-fold rotation centers determine the whole {\bf p6} lattice of rotation
centers; 
rotations about the two chosen black and white centers generate the {\bf p6} group.

The shaded triangular region shown in Fig.~\ref{p6lattice}(b) is a standard fundamental domain for the {\bf p6}
group generated by 6-fold rotations about the black circle and 3-fold rotations about the white circle. 
Hence the area of a fundamental domain for this {\bf p6} group is given by
\begin{equation}
S=x^2+y^2+xy
\end{equation}
taking the area of a unit triangle as 1. 
Since we want our $n$-iamond to be a fundamental
domain, $n = S$. Therefore,
\begin{equation}
n=1, 3, 4, 7, 9, 12, 13, 16, 19, 21, 25, 27, 28, 31, 36, 37, 39, 43, 48, 49, 49, \ldots,
\label{p6:n}
\end{equation}  
where two pairs of $(x,y)$, namely, $(5,3)$ and $(7,0)$ correspond to $n = 49$.

The action of the {\bf p6} group classifies all unit triangles into $n$ equivalence classes and we denote
the equivalence class of a unit triangle $e$ as $C(e)$. 
We construct a set $\mathscr{T}_n$ of $n$-iamonds that are
fundamental domains for the given {\bf p6} group as in section~\ref{p3:create}, 
with $n$ chosen from the list in (\ref{p6:n}). 
Fig.~\ref{p6fighi} shows the set of inequivalent $n$-iamonds in $\mathscr{T}_n$
for $n \le 9$. 
From each of these $n$-iamonds we obtain the associated {\bf p6} tiling 
by using the black circles as 6-fold rotation centers
and white circles as 3-fold rotation centers. 
Fig.~\ref{p6tiling} shows the corresponding isohedral tilings 
produced by $n$-iamonds in Fig.~\ref{p6fighi}.
\begin{figure}
\centerline{
\psfig{file=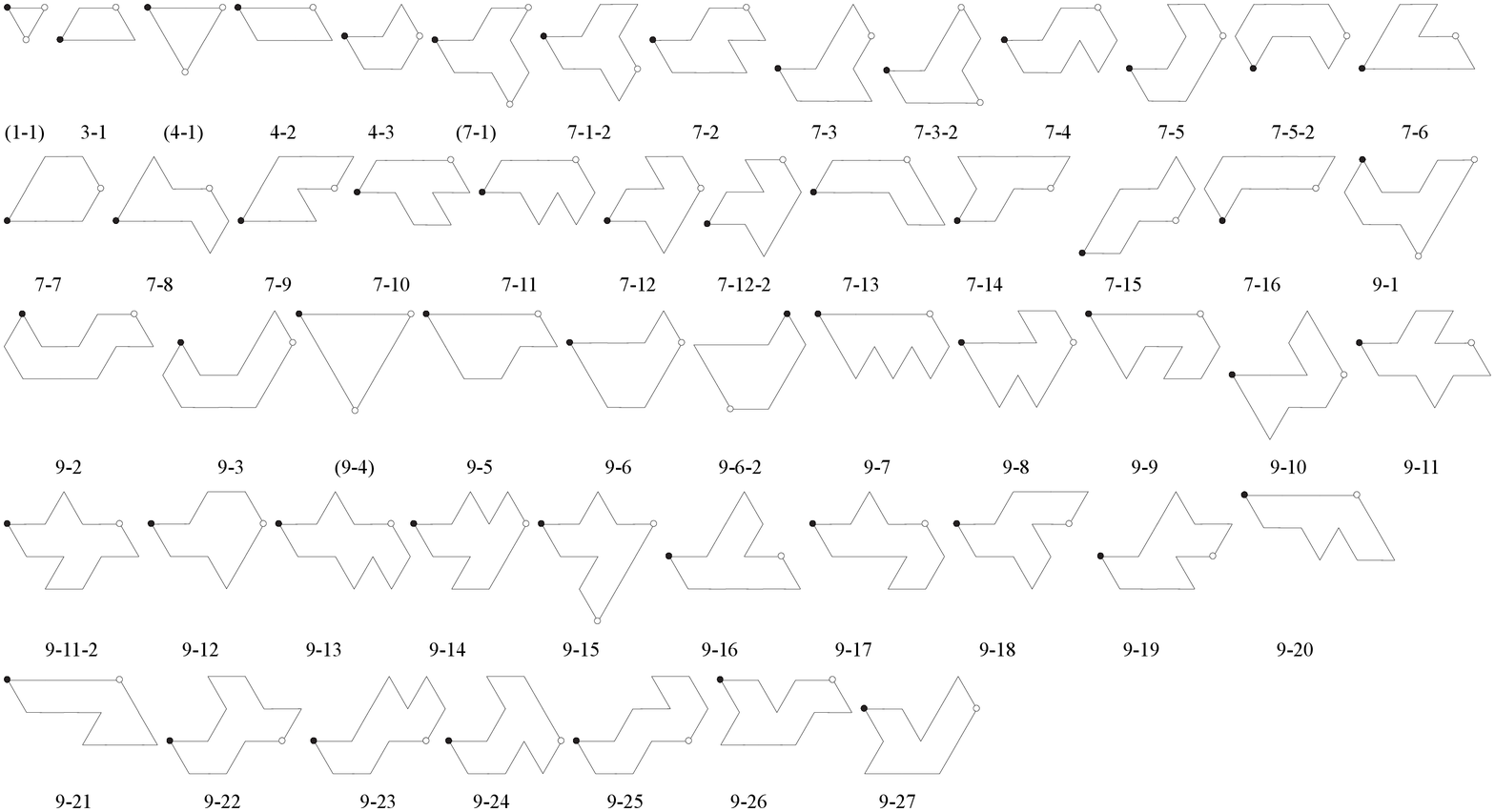,width=1.0\linewidth}}
\vspace*{8pt}
\caption{
List of $n$-iamonds produced as described in section \ref{p6:create} for $n \le 9$. 
The tiles are fundamental domains for the {\bf p6} group used to construct them. 
The labels indicate $n$ followed by the tile number for that $n$. 
Parentheses indicate that the tiles produce tilings having more symmetries than the {\bf p6} group that generates the tilings.
\label{p6fighi}
}
\end{figure}
\begin{figure}
\centerline{
\psfig{file=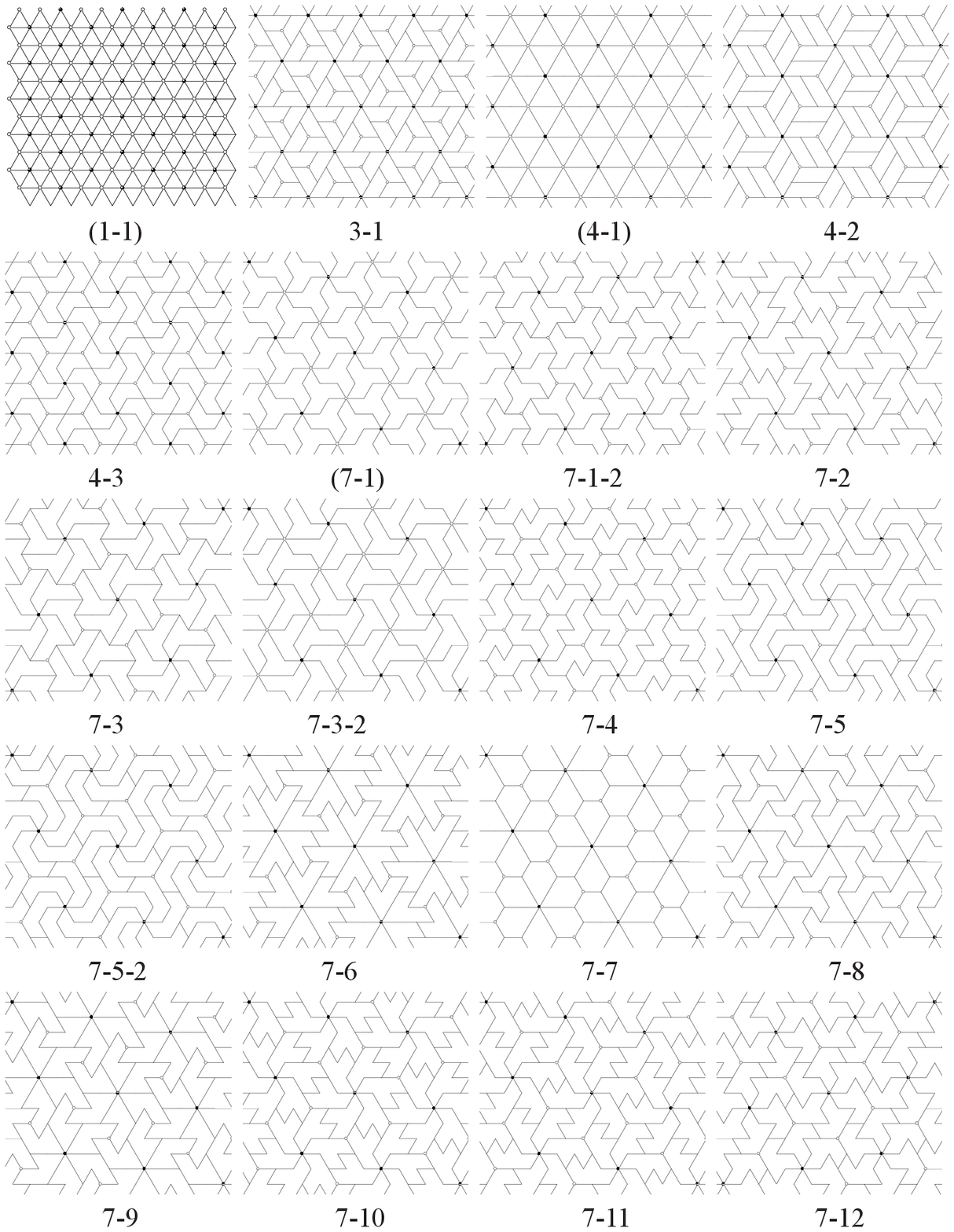,width=0.9\linewidth}}
\vspace*{8pt}
\caption{
List of {\bf p6} isohedral tilings by $n$-iamonds in Fig.~\ref{p6fighi} for $n \le 7$, 
generated by a given {\bf p6} group. 
Labels correspond to those in Fig.~\ref{p6fighi}. 
The symmetry group of each tiling is the given {\bf p6} group, 
except for tilings whose labels are in parentheses.
\label{p6tiling}
}
\end{figure}
\begin{figure}
\centerline{
\psfig{file=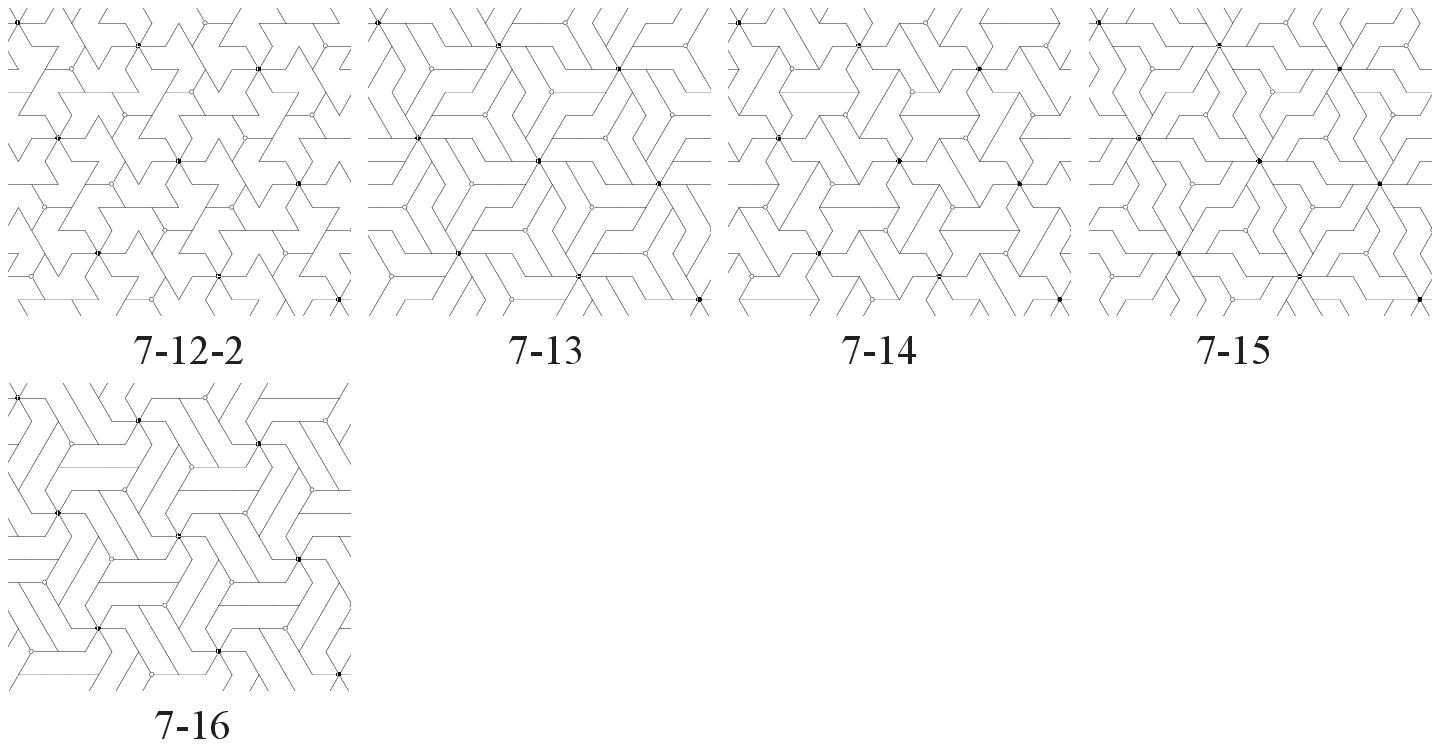,width=0.9\linewidth}}
\vspace*{8pt}
\fontsize{8pt}{0pt}\selectfont
    {\it Fig.~\ref{p6tiling}.} $(${\it Continued}$)$
\end{figure}

\subsection{Symmetries of tiles}
Some tilings in Fig.~\ref{p6tiling} have symmetries in addition to those in the {\bf p6} group $G$ 
that generated them. 
These tilings and their polyiamond tiles in Fig.~\ref{p6fighi} are indicated by parentheses 
in their labels. We can identify these polyiamonds as follows.
\begin{itemlist}
\item
Select a polyiamond that has rotation and/or reflection symmetry, and examine its tiling $\cal T$
generated by $G$.
\item
If the line connecting adjacent 6-fold and 3-fold centers (black to white) is a reflection axis
for $\cal T$, the tiling has {\bf p6m} symmetry. 
Tilings (1-1) and (4-1) in Fig.~\ref{p6tiling} have {\bf p6m} symmetry.
These tilings also have 3-fold centers for the tiling at the centers of the polyiamond tiles.
\item
Otherwise, look at all vertices and centers of unit triangles in a polyiamond in $\cal T$ except for
those black and white rotation centers in the {\bf p6} group we have generated and determine whether
or not they can be new 3-fold centers for $\cal T$. 
If new 3-fold centers are found, the full symmetry 
group of the tiling contains $G$ as a proper subgroup. 
Tiling (7-1) in Fig.~\ref{p6tiling} has a larger {\bf p6} symmetry group than $G$ 
since there are 3-fold rotation centers for $\cal T$ at the centers of the rotorlike polyiamond tiles. 
Note that the same tiles appear in tiling 7-1-2, but in this tiling, the centers
of the tiles are not 3-fold centers of rotation for $\cal T$.
\end{itemlist}

\section{{\bf p6m}}
\label{p6m}
Fig.~\ref{fig:p6m} shows the symmetry elements of a {\bf p6m} group. 
\begin{figure}[h]
\centerline{
\psfig{file=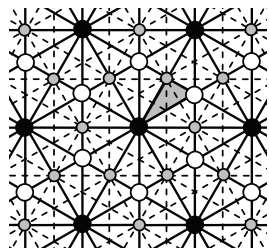,width=0.27\linewidth}}
\vspace*{8pt}
\caption{
The lattice of symmetry elements for {\bf p6m}. 
Solid lines are reflection axes and dashed lines are glide reflections axes. 
Black and white circles are 6-fold and 3-fold centers, respectively, and smaller grey circles where
two reflection axes intersect are 2-fold centers. The shaded region is a fundamental domain.
\label{fig:p6m}
}
\end{figure}
Black, white, and gray circles are the 6-, 3-, and 2-fold centers, respectively. 
Solid lines are reflection axes and dashed lines are glidereflection axes. 
The shaded $30^\circ$- $60^\circ$- $90^\circ$-triangle is a fundamental domain. 
A polyiamond that is a fundamental domain for the tiling must have its edges on the reflection axes, 
by observation (i) in section \ref{intro}. 
But this is clearly impossible. Thus,
\begin{theorem}
There are no {\bf p6m} isohedral tilings having polyiamonds as fundamental domains.
\end{theorem}

\section{Enumeration Tables}
\label{tables}
Tables 1--3 give 
the number of tiles and isohedral tilings that we have generated in sections \ref{p4}
through \ref{p6m}. 
$N_n$ is the number of inequivalent tiles $T$ in $\mathscr{T}_n$
for a given symmetry group $G$ of type 
{\bf p3}, {\bf p31m}, {\bf p3m1}, {\bf p4}, {\bf p4m}, {\bf p4g}, {\bf p6},
and {\bf p6m} that generate an isohedral tiling having
the $n$-omino or $n$-iamond tiles as fundamental domain. 
Tiles are equivalent only if they generate the same tiling by the action of the group $G$ 
when each tile is marked with an asymmetric motif. 
For example, the tiles 5-2 and 5-2-2 in Fig.~\ref{p4tiling} are congruent and their
corresponding tilings are the same, but the placement of their 4-fold centers is different, 
and so if the tiles are marked with an asymmetric motif, they generate different isohedral tilings.

Marking each tile with an asymmetric motif also guarantees that the group that generates the
tiling is the full symmetry group of the tiling. 
So $N_n$ is also the number of isohedral tilings
having $G$ as full symmetry group and having $n$-omino or $n$-iamond tiles as fundamental domain,
when each tile is marked with asymmetric motif.

$S_n$ is is the corresponding number of the $N_n$ tilings when the asymmetric motif of each tile is
removed. That is, $S_n$ is the number of isohedral tilings having full symmetry group $G$ and having
(unmarked) $n$-ominoes or $n$-iamonds as fundamental domains. 
In our figures that depict the
isohedral tilings for small values of $n$, these tilings do not have parentheses around their labels. 
These are the most important ``counting" results in this article. 
$N'_n$ is the number of non-congruent
tiles in the $N_n$ tilings, counted by ignoring rotation centers attached to tiles; similarly, 
$S'_n$ is the number of non-congruent tiles in the $S_n$ tilings.
\begin{table}[h]
\tbl{{\bf p3}, {\bf p31m} and {\bf p3m1}.}{
\begin{tabular}{lrrrr} 
\multicolumn{5}{l}{
{\bf p3} $n$-iamonds }\\
\toprule
$n$    & 2 & 6 & 8 & 14 \\
\colrule
$N_n$  & 1 & 4 & 7 & 306 \\
$S_n$  & 0 & 1 & 6 & 294 \\
$N'_n$ & 1 & 4 & 7 & 288 \\
$S'_n$ & 0 & 1 & 6 & 277 \\
\botrule
\end{tabular}
\hspace{1em}
\begin{tabular}{lrrr}
\multicolumn{4}{l}{
{\bf p31m} $n$-iamonds }\\
\toprule
$n$    & 3 & 12 \\
\colrule
$N_n$  & 1 & 20 \\
$S_n$  & 1 & 20 \\
$N'_n$ & 1 & 20 \\
$S'_n$ & 1 & 20 \\
\botrule
\end{tabular}
\hspace{1em}
\begin{tabular}{lrrrr}
\multicolumn{5}{l}{
{\bf p3m1} $n$-iamonds }\\
\toprule
$n$    & 1 & 4 & 9 & $\ldots$ \\
\colrule
$N_n$  & 1 & 1 & 1 & $\ldots$ \\
$S_n$  & 0 & 0 & 0 & $\ldots$ \\
$N'_n$ & 1 & 1 & 1 & $\ldots$ \\
$S'_n$ & 0 & 0 & 0 & $\ldots$ \\
\botrule
\end{tabular}
}
\label{tbl:p3}
\end{table}

\begin{table}[h]
\label{tbl:p4}
\tbl{{\bf p4}, {\bf p4g} and {\bf p4m}.}{
%
\begin{tabular}{lrrrrrrr}
\multicolumn{8}{l}{
{\bf p4} $n$-ominoes}\\
\toprule
$n$    & 1 & 2 & 4 & 5  & 8  & 9  & 10 \\
\colrule
$N_n$  & 1 & 1 & 3 & 12 & 45 & 82 & 300 \\
$S_n$  & 0 & 0 & 2 & 9  & 38 & 77 & 296 \\
$N'_n$ & 1 & 1 & 3 & 8  & 45 & 80 & 277 \\
$S'_n$ & 0 & 0 & 2 & 7  & 38 & 76 & 275 \\
\botrule
\end{tabular}
\hspace{1em}
\begin{tabular}{lrrrr}
\multicolumn{5}{l}{
{\bf p4g} $n$-ominoes}\\
\toprule
$n$    & 1 & 4 & 9 & 16    \\
\colrule
$N_n$  & 1 & 3 & 26 & 596  \\
$S_n$  & 0 & 2 & 25 & 595  \\
$N'_n$ & 1 & 3 & 26 & 596  \\
$S'_n$ & 0 & 2 & 25 & 595  \\
\botrule
\end{tabular}
\hspace{1em}
\begin{tabular}{lrrr}
\multicolumn{4}{l}{
{\bf p4m} $n$-ominoes}\\
\toprule
$n$    & 1 & 2 & $\ldots$ \\
\colrule
$N_n$  & 0 & 0 & $\ldots$ \\
$S_n$  & 0 & 0 & $\ldots$ \\
$N'_n$ & 0 & 0 & $\ldots$ \\
$S'_n$ & 0 & 0 & $\ldots$ \\
\botrule
\end{tabular}
}
\end{table}

\begin{table}[h]
\label{tbl6}
\tbl{{\bf p6} and {\bf p6m}.}{
%
\begin{tabular}{lrrrrrr}
\multicolumn{7}{l}{
{\bf p6} $n$-iamonds }\\
\toprule
$n$    & 1 & 3 & 4 & 7 & 9 & 12 \\
\colrule
$N_n$  & 1 & 1 & 3 & 20 & 29 & 195 \\
$S_n$  & 0 & 1 & 2 & 19 & 28 & 194 \\
$N'_n$ & 1 & 1 & 3 & 16 & 27 & 191 \\
$S'_n$ & 0 & 1 & 2 & 16 & 26 & 190 \\
\botrule
\end{tabular}
\hspace{1em}
\begin{tabular}{lrrrr}
\multicolumn{5}{l}{
{\bf p6m} $n$-iamonds }\\
\toprule
$n$    & 1 & 2 & 3 & $\ldots$ \\
\colrule
$N_n$  & 0 & 0 & 0 & $\ldots$ \\
$S_n$  & 0 & 0 & 0 & $\ldots$ \\
$N'_n$ & 0 & 0 & 0 & $\ldots$ \\
$S'_n$ & 0 & 0 & 0 & $\ldots$ \\
\botrule
\end{tabular}

}
\end{table}

For example, for symmetry group {\bf p6} and $n = 7$, 
there are 20 tiles in Fig.~\ref{p6fighi}, with
corresponding isohedral tilings in Fig.~\ref{p6tiling}, so $N_7 = 20$. 
Since tiling 7-1 of Fig.~\ref{p6tiling} has parentheses
around its label, $S_7 = 19$. 
From Fig.~\ref{p6fighi}, 
we can see that four pairs of tiles are congruent: 7-1 and
7-1-2; 7-3 and 7-3-2; 7-5 and 7-5-2; 7-12 and 7-12-2. 
Thus $N'_7 = 20 - 4 = 16$. 
Among the 19 tiles counted for $S_7$, there are also 16 non-congruent tiles, so $S'_7 = 16$.

\section{Summary}
We have described computer algorithms that can enumerate and display isohedral tilings by 
$n$-omino or $n$-iamond tiles for given $n$ in which the tiles are fundamental domains and the tilings
have 3-, 4-, or 6-fold rotational symmetry. 
Their symmetry groups are of types {\bf p3}, {\bf p31m}, {\bf p4}, {\bf p4g}, and {\bf p6}. 
We have shown that there are no isohedral tilings with symmetry groups of types
{\bf p3m1}, {\bf p4m}, or {\bf p6m} that have polyominoes or polyiamonds as fundamental domains. 
For symmetry groups of types {\bf p3}, {\bf p31m}, {\bf p4}, {\bf p4g}, 
and {\bf p6} we used the backtracking Procedure 1 to obtain a set $\mathscr{T}_n$ of $n$-omino 
or $n$-iamond tiles where each tile produced one isohedral tiling,
generated by a given symmetry group $G$ of one of these five types. 
We can denote $\mathscr{T}_{n}(G)$ as the set $\mathscr{T}_n$ for that symmetry group $G$ 
and $\mathscr{T}^{*}_{n}(G)$ the corresponding set of isohedral tilings.

We investigated the symmetries of tilings in the set $\mathscr{T}^{*}_{n}(G)$ and noted those tilings that
satisfy the following two conditions: 
(1) the full symmetry group of the tiling is $G$, and 
(2) the tiles are fundamental domains for $G$. 
We denote the subset of $\mathscr{T}^{*}_{n}(G)$ that satisfies (1) and (2)
as $\mathscr{S}^{*}_{n}(G)$ and the corresponding set of tiles as $\mathscr{S}_{n}(G)$. 
(For small values of $n$, these tiles and their tilings were displayed with labels without parentheses.) 
The enumeration of $\mathscr{S}^{*}_{n}(G)$ is the main counting result of this article. 
Although the $n$-omino or $n$-iamond tiles produced by our algorithm are not 
always fundamental domains for the isohedral tilings they generate, if we mark
these tiles with an asymmetric motif, then the set $\mathscr{T}^{*}_{n}(G)$ is the set of all isohedral tilings with
symmetry group $G$ in which the corresponding tiles in $\mathscr{T}_{n}(G)$ are fundamental domains. 
The set $\mathscr{T}^{*}_{n}(G)$ can then also include a marked fundamental domain for a {\bf p3m1} symmetry group. 
In Tables 1--3 of section \ref{tables}, we used the notation $N_n=\#\mathscr{T}^{*}_{n}(G)$ and 
$S_n=\#\mathscr{S}^{*}_{n}(G)$.


\end{document}